\documentclass[10pt,journal,epsfig]{IEEEtran}

\usepackage[dvips]{graphicx}
\usepackage{graphicx}
\usepackage{amssymb}
\usepackage{cite}
\usepackage{amsmath}
\usepackage{algorithm}
\usepackage{algorithmic}
\usepackage{multirow}
\usepackage[table]{xcolor}
\usepackage{subfigure}
\usepackage{makecell}
\usepackage{diagbox}
\usepackage{array}
\usepackage{threeparttable}
\usepackage{graphicx}
\graphicspath{ {./images/} }
\usepackage{caption,epstopdf}

\newcounter{MYtempeqncnt}
\allowdisplaybreaks
\allowbreak
\begin{document}

\title{Joint Transceiver Optimization for MmWave/THz MU-MIMO ISAC Systems}
\author{Peilan Wang, Jun Fang, Xianlong Zeng, Zhi Chen,
Hongbin Li,~\IEEEmembership{Fellow,~IEEE}
\thanks{Peilan Wang, Jun Fang, Xianlong Zeng, and Zhi Chen are with National Key
Laboratory of Wireless Communications, University of
Electronic Science and Technology of China, Chengdu 611731, China,
(Email: peilan\_wangle@uestc.edu.cn, JunFang@uestc.edu.cn, xlong\_zeng@std.uestc.edu.cn, chenzhi@uestc.edu.cn
}
\thanks{Hongbin Li is with the Department of Electrical and Computer Engineering,
Stevens Institute of Technology, Hoboken, NJ 07030, USA, E-mail:
Hongbin.Li@stevens.edu}
}

\maketitle

\begin{abstract}
In this paper, we consider the problem of joint transceiver design
for millimeter wave (mmWave)/Terahertz (THz) multi-user MIMO
integrated sensing and communication (ISAC) systems. Such a
problem is formulated into a nonconvex optimization problem, with
the objective of maximizing a weighted sum of communication users'
rates and the passive radar's signal-to-clutter-and-noise-ratio
(SCNR). By exploring a low-dimensional subspace property of the
optimal precoder, a low-complexity block-coordinate-descent
(BCD)-based algorithm is proposed. Our analysis reveals that the
hybrid analog/digital beamforming structure can attain the same
performance as that of a fully digital precoder, provided that the
number of radio frequency (RF) chains is no less than the number
of resolvable signal paths. Also, through expressing the precoder
as a sum of a communication-precoder and a sensing-precoder, we
develop an analytical solution to the joint transceiver design problem by
generalizing the idea of block-diagonalization (BD) to the ISAC
system. Simulation results show that with a proper tradeoff
parameter, the proposed methods can achieve a decent compromise
between communication and sensing, where the performance of each
communication/sensing task experiences only a mild performance loss as
compared with the performance attained by optimizing exclusively
for a single task.
\end{abstract}






\section{Introduction}
The sixth generation (6G) wireless network is envisioned to
support not only full-dimensional \emph{wireless connectivity} but
also enhanced \emph{sensing} capabilities
\cite{chen2021Terahertz,tan2022THz}. Wireless communication and
radar sensing have flourished as different disciplines due to
their different objectives. Recently, with the development of
millimeter-wave (mmWave)/terahertz (THz) communications as well as
large-scale antenna arrays, mmWave/THz communication systems and
radar sensing systems are now sharing many similarities in
hardware structures, channel characteristics and signal processing
techniques. Consequently, integrated sensing and communication
(ISAC) is emerging as a paradigm-shifting concept with a great
potential in revolutionizing both fields
\cite{liu2022Survey,liu2023Integrated}.



Joint transceiver design is a key problem in mmWave/THz ISAC
systems due to the inevitable competition for communication and
sensing resources and the subsequent interference management
\cite{liu2022Survey,xie2022Perceptive}. For mmWave and THz
communications, massive antennas are employed at the base station
(BS) and/or user equipments (UE) to compensate for the severe path
loss incurred by mmWave/THz signals. Moreover, due to hardware and
power constraints, hybrid analog and digital structures with a
small number of radio frequency (RF) chains are usually adopted
\cite{WangFang20a,GhoshThomas14,RanganRappaport14,chen2021Terahertz,tan2022THz}.
The large-scale antenna array along with a hybrid
precoder/combiner structure makes joint transceiver design for
ISAC systems a challenging problem. In addition, as the joint
transceiver design has to be performed for each interval of
channel coherence time that could be less than several
milliseconds \cite{AdhikaryAl14}, low-complexity transceiver
optimization algorithms are highly desirable in practical systems.



A plethora of studies
\cite{qi2022Hybrid,cheng2022QoSAware,wang2022PartiallyConnected,
barneto2022Beamformer,cheng2023DoublePhaseShifter,cheng2021Hybrid,
liu2019Hybrid,elbir2021TerahertzBanda,islam2022Integrated} have
been made to jointly devising transmit precoder and receive
combiner for mmWave/THz ISAC systems. The existing body of
literature can be roughly classified into three categories based
on their design criteria. Communication-oriented approaches put
their priority on communication tasks and try to improve the
sensing accuracy given that the communication performance is
guaranteed. Meanwhile, sensing-oriented approaches
\cite{cheng2023DoublePhaseShifter,cheng2021Hybrid} attempt to
optimize the communication performance given that a certain
sensing performance is attained. In addition to these approaches,
some other works
\cite{liu2019Hybrid,elbir2021TerahertzBanda,cheng2021Hybrida,islam2022Integrated}
attempt to achieve a decent tradeoff between communication and
sensing performance without explicitly considering any performance
constraints on sensing or communication. Among them,
\cite{liu2019Hybrid,elbir2021TerahertzBanda} proposed to find a
precoder for mmWave/THz ISAC systems such that the precoder is as
close as possible to the optimal precoder for communication and
also as close as possible to the optimal precoder for sensing.
Other tradeoffs between communication and sensing were also
exploited. For instance, \cite{cheng2021Hybrida} proposed to
optimize the tradeoff between the weighted sum rate (WSR) and the
radar beam pattern matching error, while
\cite{islam2022Integrated} aimed to maximize the sum of
communication and sensing signal-to-noise-ratios (SNRs).



In this paper, we consider the problem of joint transceiver design
for mmWave/THz multi-user MIMO ISAC systems, where a multi-antenna
base station (BS) serves multiple users and a passive radar
receiver is deployed to detect targets by processing reflections
from BS's downlink communication signals. Our objective is to
optimize the transmit precoder along with the receive combiner
such that the communication and sensing performance can be well
compromised. To more accurately characterize the communication and
sensing performance, a metric in terms of users' weighted sum rate
is used to evaluate the communication performance and a metric in
terms of the signal-to-clutter-and-noise-ratio (SCNR) is used to
evaluate the sensing performance. We aim at maximizing the sum of
communication users' sum rate and the radar's received SCNR
subject to a certain unit modulus and transmit power constraints.
To tackle such a non-convex optimization problem, we propose two
solutions, namely, a block-coordinate-descent (BCD)-based method
which alternatively optimizes each block variable given other
variables fixed, and an analytical solution that generalizes the
block-diagonalization (BD) idea to the ISAC systems for joint
transceiver design. Both methods enjoy a low computational
complexity that scales linearly with the number of antennas at the
BS. Specifically, a low-dimensional subspace property associated
with the optimal precoder is explored to reduce the complexity of
the BCD method. In addition, this low-dimensional subspace
property also sheds light on the minimum number of RF chains that
is required to achieve a fully-digital precoding performance for
ISAC systems.

The rest of paper is organized as follows. In Section
\ref{sec-sys}, the system model and the problem formulation are
discussed. In Section \ref{sec-BCD}, by exploring the
low-dimensional subspace property, the BCD-based method is
developed for joint transceiver design. Next, a simple analytical
joint transceiver design method is proposed in Section
\ref{sec-linear}. Simulation results are provided in Section
\ref{sec-simu}, followed by concluding remarks in Section
\ref{sec-conclu}.

\begin{figure}[!t]
\centering
\includegraphics[width=3.4in] {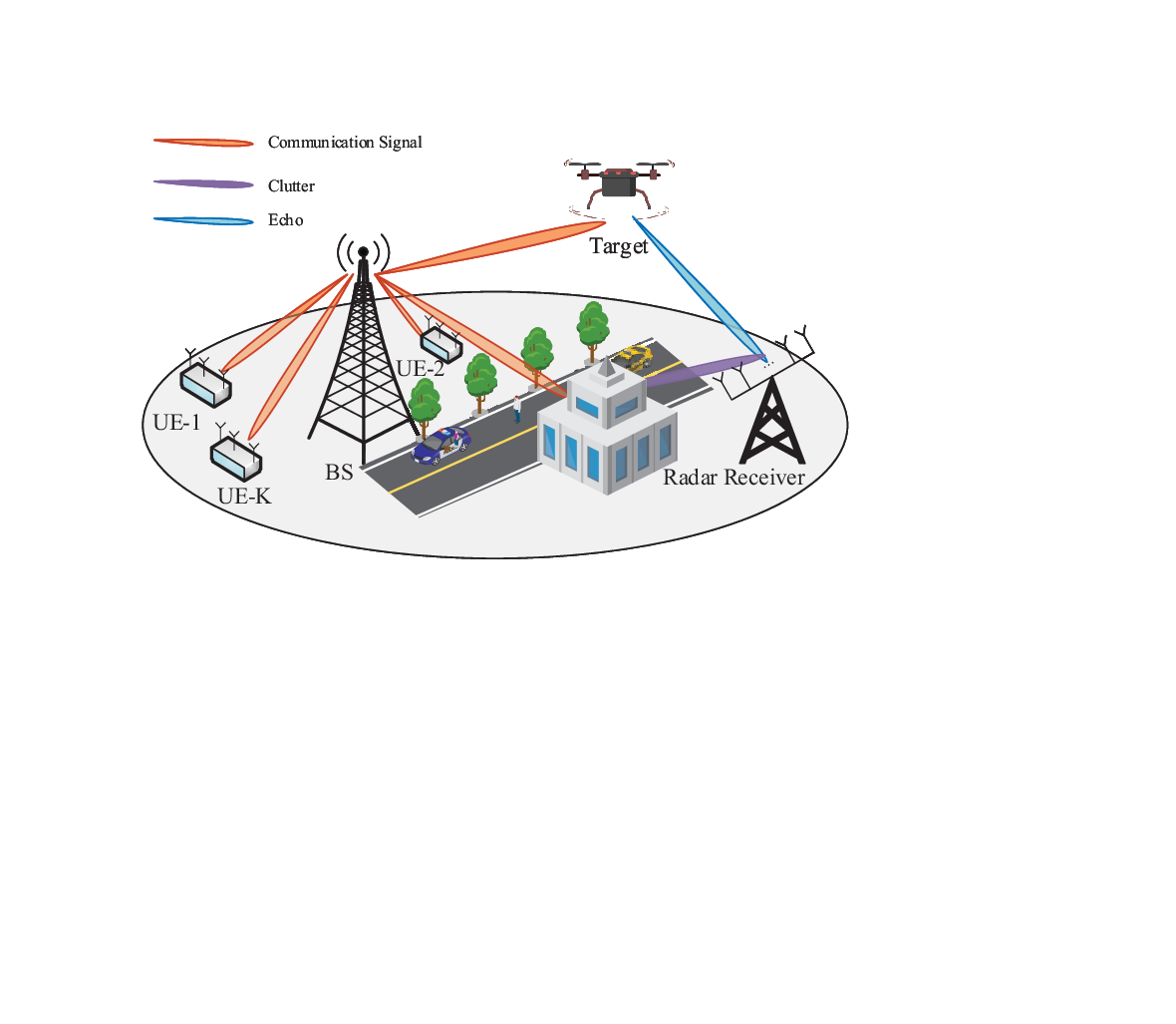}
\caption{Illustration of mmWave/THz MU-MIMO ISAC systems}
\label{System_fig}
\end{figure}

\section{System Model}\label{sec-sys}
In this paper, we consider an ISAC system where a multi-antenna BS
serves $K$ multi-antenna UEs, and a passive radar receiver is
deployed to detect a point-like target by processing reflections
from BS's downlink communication signals. To avoid
self-interference, we consider a bistatic setup where the BS and
the radar receiver are geographically separated (see Fig.
\ref{System_fig}). The BS and the radar are cable-connected such
that they can exchange necessary information to facilitate
detection and tracking of the target.


To reduce the hardware cost, we consider a hybrid analog and
digital structure at both the BS and UEs. Specifically, the BS is
equipped with $N_t$ antennas and $N_{t}^{\rm RF}$ radio frequency
(RF) chains, and each UE is equipped with $N_r$ antennas and
$N_r^{\rm RF}$ RF chains. The BS first employs a digital precoder
$ \boldsymbol{F}_{{\rm BB },k} \in \mathbb C^{N_{t}^{\rm RF}
\times N_s}$ to precode the $k$th user's transmitted symbol
$\boldsymbol{s}_k \in \mathbb C^{N_s}$, where $N_s$ denotes the
number of data streams and the transmitted symbol
$\boldsymbol{s}_k$ is assumed to satisfy $\mathbb E
[\boldsymbol{s}_k \boldsymbol{s}_k^H] = \boldsymbol{I}_{N_s}$. The
signal is then processed by an analog precoder
$\boldsymbol{F}_{\rm RF}  \in \mathbb C^{N_{t} \times N_{t}^{\rm
RF}}$ with its elements satisfying unit modulus constraints. The
baseband signal forwarded to the $k$th user can be written as
$\boldsymbol{x}_k = \boldsymbol{F}_{\rm RF} \boldsymbol{F}_{{\rm
BB },k} \boldsymbol{s}_k$, and the transmitted signal can be
expressed as
\begin{align}
\boldsymbol{x} = \sum_{k=1}^K \boldsymbol{F}_k \boldsymbol{s}_k,
\label{eq-x}
\end{align}
where $\boldsymbol{F}_k \triangleq \boldsymbol{F}_{\rm RF}
\boldsymbol{F}_{{\rm BB },k} $, and the transmit power constraint
is given by $ \mathbb E [  \| \boldsymbol{x} \|_2^2] =
\sum_{k=1}^K \text{Tr} (\boldsymbol{F}_k \boldsymbol{F}_k ^H) \leq
P_t$ with $P_t$ denoting the maximum transmit power.

\subsection{MU-MIMO Communication Performance}
The transmitted signal arrives at the $k$th user via propagating
through the channel between the BS and the $k$th UE
$\boldsymbol{H}_k \in \mathbb C^{N_r \times N_t}$. Due to the
small wavelength, mmWave/THz channels exhibit sparse scattering
characteristics and can be characterized by the Saleh-Valenzuela
(S-V) model \cite{AdhikaryAl14,chen2021Terahertz}. Suppose both the BS and the UE are equipped
with uniform linear arrays (ULA), $\boldsymbol{H}_k$ can be
modeled as
\begin{align}
\boldsymbol{H}_k = \sum_{l_k=1}^{L_k} \beta_{l_k}
\boldsymbol{a}_r(\gamma_{l_k}) \boldsymbol{a}_t^H(\varphi_{l_k}),
\end{align}
where $L_k$ denotes the total number of signal paths,
$\beta_{l_k}$ is the complex gain associated with the $l_k$th
path, $\varphi_{l_k}$ and $\gamma_{l_k}$ respectively denote the
angle-of-departure (AoD) and the angle-of-arrival (AoA) associated
with the $l_k$th path, $\boldsymbol{a}_r (\cdot)$ and
$\boldsymbol{a}_t(\cdot)$ denote the normalized receive and
transmit array steering vectors at the UE and the BS,
respectively.

The signal received by the $k$th UE is given by
\begin{align}
\bar{\boldsymbol{y}}_k = \boldsymbol{H}_k \boldsymbol{F}_k
\boldsymbol{s}_k + \sum_{j\neq k}^K  \boldsymbol{H}_k
\boldsymbol{F}_{j} \boldsymbol{s}_j +  \boldsymbol{n}_k,
\end{align}
where $\boldsymbol{n}_k \sim \mathcal{CN}(0,\sigma_k^2
\boldsymbol{I})$ is the additive white Gaussian noise (AWGN). The
received signal is first processed by an analog combiner
$\boldsymbol{W}_{{\rm RF},k} \in \mathbb C^{N_r \times N_{r}^{\rm
RF}}$ and then followed by a lower-dimensional digital combiner $
\boldsymbol{W}_{{\rm BB},k} \in \mathbb C^{N_{r}^{\rm RF} \times
N_s}$. Define $\boldsymbol{W}_k\triangleq\boldsymbol{W}_{{\rm RF},k}
\boldsymbol{W}_{{\rm BB},k} \in \mathbb C^{N_r \times N_s}$, the
received baseband signal is thus given as
\begin{align}
\boldsymbol{y}_k = \boldsymbol{W}_k^H \boldsymbol{H}_k
\boldsymbol{F}_k \boldsymbol{s}_k + \sum_{j\neq k}^K
\boldsymbol{W}_k^H \boldsymbol{H}_k \boldsymbol{F}_{j}
\boldsymbol{s}_j + \boldsymbol{W}_k^H \boldsymbol{n}_k.
\end{align}
Accordingly, the achievable rate of the $k$th user is given by
\begin{align}
R_k = \log_2 {\rm det} \left( \boldsymbol{I} + \boldsymbol{W}_k^H
\boldsymbol{H}_k \boldsymbol{F}_k \boldsymbol{F}_k^H
\boldsymbol{H}_k^H \boldsymbol{W}_k \boldsymbol{J}_k^{-1}
\right),\label{eq-Rk}
\end{align}
where
\begin{align}
\boldsymbol{J}_k \triangleq\sum_{j\neq k}^K
\boldsymbol{W}_k^H\boldsymbol{H}_k \boldsymbol{F}_j
\boldsymbol{F}_j^H \boldsymbol{H}_k^H \boldsymbol{W}_k +
\sigma_k^2 \boldsymbol{W}_k^H\boldsymbol{W}_k. \label{eq-Jk}
\end{align}
The WSR can be calculated as
\begin{align}
    R ( \boldsymbol{F}_k,\boldsymbol{W}_k)= \sum_{k=1}^K w_k R_k, \label{eq-R}
\end{align}
where $w_k$ denotes a weight characterizing the priority of the
$k$th UE. Note that the system parameters satisfy $KN_s \leq
N_{t}^{\rm RF} \leq N_{t}$ and $N_s \leq N_r^{\rm RF} \leq N_r$.

\subsection{MIMO Sensing Performance}
The radar receiver is equipped with $N_{\rm sen}$ antennas.
Assuming a point-like target is located at angle $\theta_0$ with
respect to (w.r.t.) the BS and angle $\varphi_0$ w.r.t the radar
receiver,  the target response matrix is expressed as
\begin{align}
\boldsymbol{H}_A = \alpha_0 \boldsymbol{a}_{s}(\varphi_0)
\boldsymbol{a}_t^H(\theta_0),
\end{align}
where $\boldsymbol{a}_s(\cdot) \in \mathbb C^{N_{\rm sen}}$
denotes the normalized receive array steering vector at the radar
receiver, $\alpha_0$ is the complex gain incorporating the target
radar cross section (RCS) as well as the path loss. Here,
$\alpha_0$ is assumed to be a zero mean complex Gaussian random
variable with variance $\sigma_A^2$ \cite{xie2022Perceptive}.  In addition to the
echo back-scattered from the target, the radar receiver also
receives echoes from the environment \cite{xie2022Perceptive,chen2022generalized}, i.e., the clutter.
Suppose there are $I$ clutter patches \footnote{The line-of-sight
(LOS) path from the BS to the radar receiver can also be regarded
as a strong clutter patch.}. The response matrix associated with
the $i$th clutter is given as
\begin{align}
\boldsymbol{H}_{B,i} = \alpha_i  \boldsymbol{a}_s(\varphi_i)
\boldsymbol{a}_t^H(\theta_i),
\end{align}
where $\alpha_{i} \sim \mathcal{CN}(0,\sigma_{B_i}^2)$ is the
complex gain associated with the $i$th clutter patch,
$\varphi_i$ and $\theta_i$ denote the AoA and AoD at the radar
receiver and BS, respectively. The received signal at the radar
receiver is given by
\begin{align}
\boldsymbol{y}_0 =&  \boldsymbol {H}_{A} \boldsymbol{x} +
\sum_{i=1}^{I}\boldsymbol{H}_{B,i} \boldsymbol{x}  +
\boldsymbol{n}_0,
\end{align}
where $\boldsymbol{x}$ is the communication signal defined in
\eqref{eq-x}, and $\boldsymbol{n}_0 \in \mathbb C^{N_{\rm sen}
\times 1} $ is the AWGN satisfying $\mathcal{CN}(0, \sigma^2
\boldsymbol{I})$.

Let $\boldsymbol{w} \in \mathbb C^{N_{\rm sen}}$ denote the radar
receive beamforming vector. The output of the radar receiver is
given by
\begin{align}
r = \boldsymbol{w}^H \boldsymbol{y}_0 = \boldsymbol{w}^H
\boldsymbol {H}_{A} \boldsymbol{x} +  \boldsymbol{w}^H
\sum_{i=1}^{I}\boldsymbol{H}_{B,i} \boldsymbol{x}  +
\boldsymbol{w}^H \boldsymbol{n}_0.
\end{align}
For sensing tasks, the SCNR
determines both detection and localization performance \cite{xie2022Perceptive,chen2022generalized}.
Therefore we use SCNR as a metric to evaluate the sensing
performance. Specifically, the SCNR is defined as the ratio of the
average received signal power to the average received clutter
power plus the noise power \cite{xie2022Perceptive,cheng2023DoublePhaseShifter}, i.e.,
\begin{align}
{\text{SCNR}} (\boldsymbol{w}, \boldsymbol{F}_k) \stackrel{(a)}=&
\frac{\mathbb E[ |\boldsymbol{w}^H \boldsymbol {H}_A \boldsymbol{x} |^2]}
{\mathbb E[|\boldsymbol{w}^H \sum_{i=1}^I\boldsymbol{H}_{B,i} \boldsymbol{x}
+ \boldsymbol{w}^H \boldsymbol{n}_0|^2]} \nonumber \\
=& \frac{\boldsymbol{w}^H  \boldsymbol {A}  \boldsymbol{F}
\boldsymbol {F}^H  \boldsymbol{A}^H \boldsymbol{w}}{
\boldsymbol{w}^H \left(  \sum_{i=1}^I \boldsymbol{B}_i \boldsymbol
{F} \boldsymbol {F}^H \boldsymbol {B}_i^H + \sigma^2
\boldsymbol{I}   \right) \boldsymbol{w}  } , \label{eq-SCNR}
\end{align}
where the expectation in $(a)$ is taken over the transmitted
signal $\boldsymbol{x}$ as well as the random target/clutter
response matrices, and $\boldsymbol{A}$, $\boldsymbol{B}_i$ and
$\boldsymbol{F}$ are respectively defined as
\begin{align}
\boldsymbol{A} = &\sigma_A \boldsymbol{a}_{s}(\varphi_0) \boldsymbol{a}_t^H(\theta_0), \\
\boldsymbol{B}_i = &\sigma_{B_i}  \boldsymbol{a}_s(\varphi_i) \boldsymbol{a}_t^H(\theta_i) , \forall i\\
\boldsymbol{F} =& [\boldsymbol{F}_1 \phantom{0} \boldsymbol{F}_2
\ldots \boldsymbol{F}_K] \in \mathbb C^{N_t \times KN_s}.
\end{align}
Note that the statistical information of the target and the
clutter patches $\{\sigma_A^2,\sigma_{B,i}^2\}$ are assumed known
\emph{a priori}. In practice, they can be estimated based on a
cognitive paradigm \cite{karbasi2015knowledge}.

\subsection{Problem Formulation}
As discussed earlier, the communication performance is measured by
the WSR while the sensing performance is determined by the SCNR.
In this paper, we hope to achieve a tradeoff between the
communication and sensing performance by jointly optimizing the
hybrid precoder/combiner $\{\boldsymbol{F}_k,\boldsymbol{W}_k\}$
and the radar receive beamforming vector $\boldsymbol{w}$. The
optimization problem can be formulated as
\begin{align}
\max_{ \{\boldsymbol{F}_{k},\boldsymbol{W}_{k}\}_{k=1}^K,\boldsymbol{w}}
\quad & \rho_c R( \boldsymbol{F}_k,\boldsymbol{W}_k) +
\rho_s\text{SCNR}  (\boldsymbol{w}, \boldsymbol{F}_k) \nonumber \\
{\text{s.t.}} \quad  &   \sum_{k=1}^K {\rm tr}(\boldsymbol{F}_k \boldsymbol{F}_k^H) \leq P_t, \nonumber \\
& \boldsymbol{F}_{k} = \boldsymbol{F}_{\rm RF} \boldsymbol{F}_{{\rm BB },k}, \nonumber\\
& \boldsymbol{W}_{k} = \boldsymbol{W}_{{\rm RF},k} \boldsymbol{W}_{{\rm BB },k}, \nonumber\\
 & |(\boldsymbol{F}_{\rm RF})_{i,j}| = 1, \nonumber\\
 & |(\boldsymbol{W}_{\rm RF})_{i,j}| = 1, \label{eq-opt-1}
\end{align}
where $R(\boldsymbol{F}_k,\boldsymbol{W}_k) $ is defined in
\eqref{eq-R}, $ {\text{SCNR}} (\boldsymbol{w}, \boldsymbol{F}_k)$
is defined in \eqref{eq-SCNR}. The weighting coefficients $\rho_c
= \frac{\eta}{\text{cons}_1} , \rho_s = \frac{1-\eta}{
\text{cons}_2}$ are used to control the tradeoff between
communication and sensing performance, in which $\text{cons}_1$
and $\text{cons}_2$ are predefined constants used to reduce the
weighting effects caused by the large difference in values between
SCNR and WSR. When $\eta =1$, the system focuses only on the
communication performance. On the contrary, when $\eta = 0$, the
sensing performance is exclusively optimized. As $\eta$ varies
from 0 to 1, the performance tradeoff can be characterized by the
optimal solution of \eqref{eq-opt-1}.


The optimization problem \eqref{eq-opt-1} is challenging due to
its nonconvexity as well as coupling among optimization variables.
In the following, we propose a low-complexity block-coordinated
descent-based method to solve \eqref{eq-opt-1}.

\section{Proposed BCD-Based Method}\label{sec-BCD}
For simplicity, we first ignore the constraints placed by the
hybrid structure at the transceiver and consider a fully digital
precoder/combiner $\boldsymbol{F}_k$/$\boldsymbol{W}_k$. After the
optimal fully digital precoder/combiner is obtained, an efficient
manifold optimization-based algorithm is employed to find hybrid
precoder/combiner to approximate the optimal digital
precoder/combiner.

\subsection{Problem Reformulation}
By ignoring the constraint imposed by hybrid structures, the
problem \eqref{eq-opt-1} is simplified as
\begin{align}
\max_{ \{\boldsymbol{F}_{k}  ,\boldsymbol{W}_{k}\}_{k=1}^K,\boldsymbol{w}}
\quad & \rho_c R ( \boldsymbol{F}_k,\boldsymbol{W}_k)   +
\rho_s \text{SCNR}(\boldsymbol{w}, \boldsymbol{F}_k) \nonumber \\
{\text{s.t.}} \quad  &   \sum_{k=1}^K {\rm tr}(\boldsymbol{F}_k
\boldsymbol{F}_k^H) \leq P_t. \label{eq-opt-simp0}
\end{align}
Before proceeding, we consider the problem of optimizing
$\{\boldsymbol{F}_k \}_{k=1}^K$, given fixed
$\{\boldsymbol{W}\}_{k=1}^K$ and $\boldsymbol{w}$, in which case
the problem becomes
\begin{align}
\max_{ \{\boldsymbol{F}_{k}\}_{k=1}^K} \quad &
\rho_c R ( \boldsymbol{F}_k,\boldsymbol{W}_k)   +
\rho_s \text{SCNR}  (\boldsymbol{w}, \boldsymbol{F}_k) \nonumber \\
{\text{s.t.}} \quad  &   \sum_{k=1}^K {\rm tr}(\boldsymbol{F}_k
\boldsymbol{F}_k^H) \leq P_t. \label{eq-opt-simp}
\end{align}
We first obtain an interesting and useful result regarding the
low-dimensional subspace property of $\{ \boldsymbol{F}_k
\}_{k=1}^K$ by exploiting the sparse characteristics of mmWave/THz
channels.

\newtheorem{Proposition}{Proposition}
\begin{Proposition}
Denote
\begin{align}
\boldsymbol{V} = & [ \boldsymbol{a}_t(\varphi_{1}) \ldots
\boldsymbol{a}_t(\varphi_{L_K}) \phantom{0}
\boldsymbol{a}_t(\theta_0) \phantom{0}
\boldsymbol{a}_t(\theta_1)\ldots \boldsymbol{a}_t(\theta_I)]  \in
\mathbb C^{N_t \times r}
\end{align}
as the matrix consisting of the steering vectors from the BS to all UEs,
the steering vectors from the BS to the target and the steering
vectors from the BS to all clutter patches, where we have
$r\triangleq\sum_{k=1}^K L_k + I +1 $. Then, any non-trivial
Karush-Kuhn-Tucker (KKT) point of \eqref{eq-opt-simp} can be
expressed as
\begin{align}
\boldsymbol{F}_k = \boldsymbol{V} \boldsymbol{X}_k, \quad \forall
k \label{Fk-Xk}
\end{align}
where $\boldsymbol{X}_k  \in \mathbb C^{r \times N_{s}}$.
\label{proposition1}
\end{Proposition}
\begin{proof}
See Appendix \ref{appendixA}.
\end{proof}

Define
\begin{align}
\boldsymbol{H}\triangleq& [\boldsymbol{H}_1^T\phantom{0}
\ldots\phantom{0} \boldsymbol{H}_K^T\phantom{0}
\boldsymbol{A}^T\phantom{0} \boldsymbol{B}_1^T\phantom{0}
\ldots\phantom{0}\boldsymbol{B}_{I}^T]^T \in \mathbb C^{M \times
N_t},
\end{align}
where $M=KN_r + (I+1)N_{\rm sen}$. It can be readily verified that
$\boldsymbol{V}$ and $\boldsymbol{H}^H$ have the same range space,
i.e., $\mathcal{R}(\boldsymbol{V}) = \mathcal{R}(\boldsymbol{H}^H)
$. The low-dimensional subspace property for communication systems
has been verified by classical linear precoders \cite{zhao2023Rethinking}, e.g.,  the
maximum-ratio-transmission (MRT) precoder $\boldsymbol{F}_{\rm
MRT} = \boldsymbol{H}^H $ and the zero-forcing (ZF) precoder
$\boldsymbol{F}_{\rm ZF} = \boldsymbol{H}^H
(\boldsymbol{H}\boldsymbol{H}^H)^{-1}$. Here, we show that the
precoder $\boldsymbol{F}$ for ISAC systems also exhibits this
intriguing low-dimensional subspace property.

Define $\bar{\boldsymbol{H}}_k \triangleq \boldsymbol{H}_k
\boldsymbol{V} \in \mathbb C^{N_r \times r}$,
$\bar{\boldsymbol{A}} \triangleq \boldsymbol{A}\boldsymbol{V} \in
\mathbb C^{N_{\rm sen} \times r}$, $\bar{\boldsymbol{B}}_i
\triangleq \boldsymbol{B}_i \boldsymbol{V} \in \mathbb C^{N_{\rm
sen} \times r}$, $\tilde{\boldsymbol{H}} \triangleq
\boldsymbol{V}^H\boldsymbol{V} \in \mathbb C^{r \times r}$. By
invoking Proposition \ref{proposition1}, the problem
\eqref{eq-opt-simp0} can be rewritten as
\begin{align}
\max_{ \{\boldsymbol{X}_{k}  ,\boldsymbol{W}_{k}\}_{k=1}^K,\boldsymbol{w}}
\quad & \rho_c R ( \boldsymbol{X}_k,\boldsymbol{W}_k)   +
\rho_s \text{SCNR}  (\boldsymbol{w}, \boldsymbol{X}_k) \nonumber \\
{\text{s.t.}} \quad  &   \sum_{k=1}^K {\rm
tr}(\tilde{\boldsymbol{H}}\boldsymbol{X}_k \boldsymbol{X}_k^H)
\leq P_t. \label{eq-opt-X}
\end{align}
where
\begin{align}
&R ( \boldsymbol{X}_k,\boldsymbol{W}_k) \nonumber \\
=  &\sum_{k=1}^K w_k\log_2 {\rm det} \bigg( \boldsymbol{I} +
\boldsymbol{W}_k^H \bar{\boldsymbol{H}}_k \boldsymbol{X}_k
\boldsymbol{X}_k^H \bar{\boldsymbol{H}}_k^H \boldsymbol{W}_k \nonumber \\
& \times \big( \sum_{j\neq k}^K
\boldsymbol{W}_k^H\bar{\boldsymbol{H}}_k \boldsymbol{X}_j
\boldsymbol{X}_j^H \bar{\boldsymbol{H}}_k^H \boldsymbol{W}_k +
\sigma_k^2  \boldsymbol{W}_k^H\boldsymbol{W}_k \big)^{-1} \bigg),
\end{align}
and
\begin{align}
{\text{SCNR}}(\boldsymbol{w}, \boldsymbol{X}_k) =
\frac{\boldsymbol{w}^H  \bar{\boldsymbol {A}}  \boldsymbol{X}
\boldsymbol {X}^H  \bar{\boldsymbol{A}}^H \boldsymbol{w}}{
\boldsymbol{w}^H \left( \sum_{i=1}^I \bar{\boldsymbol{B}}_i
{\boldsymbol {X}} \boldsymbol {X}^H \bar{\boldsymbol {B}}_i^H +
\sigma^2 \boldsymbol{I} \right) \boldsymbol{w}  }
\end{align}
in which $\boldsymbol{X}\triangleq
[\boldsymbol{X}_1\phantom{0}\ldots\phantom{0}\boldsymbol{X}_K] \in
\mathbb C^{r\times KN_s}$.

In \eqref{eq-opt-X}, the optimization variable $\boldsymbol{F} \in
\mathbb C^{N_t \times KN_s}$ is replaced by $\boldsymbol{X} \in
\mathbb C^{r \times KN_s}$. Note that, due to limited scattering
characteristics of mmWave/THz signals, the dimension $r$ is
usually smaller than the number of antennas $N_t$ at the BS. As a
result, Proposition \ref{proposition1} allows us to remarkably
reduce the complexity of the proposed algorithm, as will be
elaborated in Section \ref{sec-prop-BCd}.

Furthermore, it should be noted that the optimal solution to
\eqref{eq-opt-X} will always be on the boundary of the quadratic
constraint, which is known as the \emph{full power property}.
Using this property, the constraint in \eqref{eq-opt-X} can be
removed and absorbed into the objective function, which yields
\begin{align}
\max_{ \{\boldsymbol{X}_{k}
,\boldsymbol{W}_{k}\}_{k=1}^K,\boldsymbol{w}} \quad & \rho_c
\bar{R} ( \boldsymbol{X}_k,\boldsymbol{W}_k)   +  \rho_s
\overline{\text{SCNR} } (\boldsymbol{w}, \boldsymbol{X}_k),
\label{eq-opt-unc}
\end{align}
where
\begin{align}
& \bar{R} ( \boldsymbol{X}_k,\boldsymbol{W}_k) \nonumber \\ =&
\sum_{k=1}^K w_k\log_2 {\rm det} \bigg( \boldsymbol{I} +
\boldsymbol{W}_k^H \bar{\boldsymbol{H}}_k \boldsymbol{X}_k
\boldsymbol{X}_k^H \bar{\boldsymbol{H}}_k^H \boldsymbol{W}_k \nonumber \\
& \times \bigg( \sum_{j\neq k}^K \boldsymbol{W}_k^H\bar{\boldsymbol{H}}_k
\boldsymbol{X}_j \boldsymbol{X}_j^H \bar{\boldsymbol{H}}_k^H \boldsymbol{W}_k \nonumber \\
&+ \frac{\sigma_k^2}{P_t}  \sum_{i=1}^K{\rm
tr}(\tilde{\boldsymbol{H}} \boldsymbol{X}_i \boldsymbol{X}_i^H)
\boldsymbol{W}_k^H\boldsymbol{W}_k \bigg)^{-1} \bigg),
\label{eq-R-X}
\end{align}
and
\begin{align}
& \overline{\text{SCNR} }(\boldsymbol{w}, \boldsymbol{X}_k) \nonumber \\
= & \frac{\boldsymbol{w}^H  \bar{\boldsymbol {A}}  \boldsymbol{X}
\boldsymbol {X}^H  \bar{\boldsymbol{A}}^H \boldsymbol{w}}{
\boldsymbol{w}^H \left(\sum_{i=1}^I( \bar{\boldsymbol{B}}_i
{\boldsymbol {X}} \boldsymbol {X}^H \bar{\boldsymbol {B}}_i^H) +
\frac{ \sigma^2}{P_t}  \sum_{i=1}^K{\rm
tr}(\tilde{\boldsymbol{H}}\boldsymbol{X}_i \boldsymbol{X}_i^H)
\boldsymbol{I}_r   \right) \boldsymbol{w}  }.
\end{align}
Next, we show that the optimal solution of \eqref{eq-opt-X} can be
obtained by solving \eqref{eq-opt-unc}.

\begin{Proposition}
\label{proposition-scaling} Denote $(\{\boldsymbol{W}_k^{\rm
opt_1},\boldsymbol{X}_k^{\rm opt_1} \}_{k=1}^K,\boldsymbol{w}^{\rm
opt_1})$ as the optimal solution to \eqref{eq-opt-X} and
$(\{\boldsymbol{W}_k^{\rm opt_2},\boldsymbol{X}_k^{\rm opt_2}
\}_{k=1}^K,\boldsymbol{w}^{\rm opt_2})$ as the optimal solution to
\eqref{eq-opt-unc}. Then, we have
\begin{align}
\boldsymbol{W}_k^{\rm opt_1} = \boldsymbol{W}_k^{\rm opt_2},
\boldsymbol{w}^{\rm opt_1} = \boldsymbol{w}^{\rm opt_2},
\boldsymbol{X}_k^{\rm opt_1} = \alpha
\boldsymbol{X}_k^{\rm opt_2},\forall k,
\end{align}
where
\begin{align}
\alpha = \sqrt{ \frac{P_t}{\sum_{k=1}^K {\rm
tr}(\tilde{\boldsymbol{H}}\boldsymbol{X}_k^{\rm opt_2}
(\boldsymbol{X}_k^{\rm opt_2})^H)      }   }.
\end{align}
\end{Proposition}
\begin{proof}
See Appendix \ref{appendix-scaling}.
\end{proof}

So far we have converted the problem \eqref{eq-opt-simp0} into a
low-dimensional unconstrained problem \eqref{eq-opt-unc}.
Nevertheless, it is still difficult to solve \eqref{eq-opt-unc}
due to coupling among optimization variables. To address this
issue, we use the following tricks to simplify the WSR term and
the SCNR term. Specifically, by resorting to the WMMSE technique\cite{zhao2023Rethinking},
the WSR term $\bar{R} ( \boldsymbol{X}_k,\boldsymbol{W}_k) $ in
\eqref{eq-R-X} can be formulated into the following optimization:
\begin{align}
& \bar{R} ( \boldsymbol{X}_k,\boldsymbol{W}_k) \nonumber \\
=& \max_{\boldsymbol{\Lambda}_k \succ 0, \boldsymbol{E}_{k}} \quad
\sum_{k=1}^K w_k \bigg( \log_2 {\rm det}( \boldsymbol{\Lambda}_k)
- {\rm tr} ( \boldsymbol{\Lambda}_k \boldsymbol{E}_k ) \bigg) +
KN_s, \label{BCD-prop}
\end{align}
where $\{ \boldsymbol{\Lambda}_k \}_{k=1}^K$ are auxiliary
variables and
\begin{align}
\boldsymbol{E}_{k}
=& (\boldsymbol{I} - \boldsymbol{W}_k^H\bar{ \boldsymbol{H}}_k \boldsymbol{X}_k)
(\boldsymbol{I} - \boldsymbol{W}_k^H\bar{ \boldsymbol{H}}_k \boldsymbol{X}_k) ^H \nonumber \\
&+ \boldsymbol{W}_k \bar{\boldsymbol{J}_k} \boldsymbol{W}_k^H, \label{eq-Ek}
\end{align}
in which
\begin{align}
\bar{\boldsymbol{J}_k} =& \sum_{j=1,j\neq k}^K  \bar{\boldsymbol{H}}_k
\boldsymbol{X}_j \boldsymbol{X}_j^H \bar{\boldsymbol{H}}_k^H + \frac{\sigma_k^2}{P_t} \sum_{i=1}^K{\rm tr}(
\tilde{\boldsymbol{H}} \boldsymbol{X}_i \boldsymbol{X}_i^H)
\boldsymbol{I}. \label{eq-mse-2}
\end{align}

On the other hand, the SCNR term can be written as
\begin{align}
&\frac{\boldsymbol{w}^H \bar{\boldsymbol {A} } \boldsymbol{X}
\boldsymbol {X}^H  \bar{\boldsymbol {A} } \boldsymbol{w}}{
\boldsymbol{w}^H \left( \sum_{i=1}^I( \bar{\boldsymbol{B}}_i {\boldsymbol {X}}
\boldsymbol {X}^H \bar{\boldsymbol {B}}_i^H) + \frac{\sigma^2}{P_t}
\sum_{i=1}^K{\rm tr}( \tilde{\boldsymbol{H}} \boldsymbol{X}_i
\boldsymbol{X}_i^H)\boldsymbol{I}   \right) \boldsymbol{w}  } \nonumber \\
\triangleq &\frac{   f_1 ({\boldsymbol{X}}, \boldsymbol{w} )}{f_2
({\boldsymbol{X}}, \boldsymbol{w})}. \label{eq-opt-SCNR}
\end{align}
By introducing an auxiliary vector $\boldsymbol{u}
=[\boldsymbol{u}_1^T \phantom{0} \boldsymbol{u}_2^T \phantom{0}
\ldots \phantom{0} \boldsymbol{u}_K^T]^T \in \mathbb C^{KN_s}$
with $\boldsymbol{u}_k \in \mathbb C^{N_s}$, the SCNR term is
equivalent to solving the following optimization \cite{ShenYu18}
\begin{align}
\overline{\text{SCNR} } (\boldsymbol{w}, \boldsymbol{X}_k) =\max_{
\boldsymbol{u}} \; 2\Re \{ \boldsymbol{w}^H
\bar{\boldsymbol{A}} \boldsymbol{X} \boldsymbol{u} \} - \|
\boldsymbol{u}\|_2^2  f_2( {\boldsymbol{X}}, \boldsymbol{w}) .
\label{SCRN-prop}
\end{align}
Combining \eqref{BCD-prop} and \eqref{SCRN-prop}, the problem
\eqref{eq-opt-unc} is finally simplified as
\begin{align}
\min_{\{ \boldsymbol{\Lambda}_k,\boldsymbol{W}_k,
\boldsymbol{X}_k\}_{k=1}^K, \boldsymbol{w},\boldsymbol{u}} \; & \sum_{k=1}^K
\rho_c w_k \left({\rm tr} \left(\boldsymbol{\Lambda}_k
\boldsymbol{E}_{k}\right) - \log {\rm det}
(\boldsymbol{\Lambda}_k) \right) \nonumber \\ & -\rho_s \left(
2\Re \{ \boldsymbol{w}^H \bar{\boldsymbol{A}} \boldsymbol{X}
\boldsymbol{u} \} - \| \boldsymbol{u}\|_2^2  f_2( \boldsymbol{X},
\boldsymbol{w})\right). \label{eq-opt-final}
\end{align}
Although the objective function of \eqref{eq-opt-final} is
nonconvex, it is convex over each individual optimization variable
when the other variables are fixed. Hence, a
block-coordinate-descent (BCD) method can be applied to solve the
problem \eqref{eq-opt-final}.

\subsection{The Proposed BCD-Based Algorithm}
\label{sec-prop-BCd} In the sequel, we update each block variable
while fixing the others. Since each subproblem is convex, its
optimal solution can be obtained in a closed-form by checking the
first-order optimality condition. The details are given as
follows.

\subsubsection{Update $\boldsymbol{w}$}
Given other variables fixed, the optimization of $\boldsymbol{w}$
can be formulated as
\begin{align}
    \min_{\boldsymbol{w}} \quad -
2\Re \{ \boldsymbol{w}^H \boldsymbol{A} \boldsymbol{X}
\boldsymbol{u} \} + \| \boldsymbol{u}\|_2^2  f_2( \boldsymbol{X},
\boldsymbol{w}),
\end{align}
where $f_2( \boldsymbol{X}, \boldsymbol{w})$ is defined in
\eqref{eq-opt-SCNR}. By checking its first-order optimality
condition, we can update $\boldsymbol{w}$ via
\begin{align}
\boldsymbol{w} =&  \frac{1}{\| \boldsymbol{u}\|_2^2 }\left(
\sum_{i=1}^I( \bar{\boldsymbol{B}}_i {\boldsymbol {X}} \boldsymbol
{X}^H \bar{\boldsymbol {B}}_i^H)+ \frac{\sigma^2}{P_t}\sum_{k=1}^K
{\rm tr}( \tilde{\boldsymbol{H}}\boldsymbol{X}_k
\boldsymbol{X}_k^H) \boldsymbol{I}   \right)^{-1} \nonumber \\ &
\times \bar{\boldsymbol{A} } \boldsymbol{X} \boldsymbol{u}.
\label{eq-ite-w}
\end{align}

\subsubsection{Update $\boldsymbol{W}_k$}
Given other variables fixed, the optimization of
$\boldsymbol{W}_k$ is simplified as
\begin{align}
    \min_{\boldsymbol{W}_k} \quad {\rm tr}(\boldsymbol{\Lambda}_k \boldsymbol{E}_k),
\end{align}
where $\boldsymbol{E}_k$ is defined in \eqref{eq-Ek}.
Hence, the update of $\boldsymbol{W}_k$ is given by
\begin{align}
\boldsymbol{W}_k =& {\left( \sum_{i=1}^K \bar{\boldsymbol{H}}_k
\boldsymbol{X}_i \boldsymbol{X}_i^H \bar{\boldsymbol{H}}_k^H +
\frac{\sigma_k^2}{P_t}  \sum_{i=1}^K{\rm tr}( \tilde{
\boldsymbol{H}}\boldsymbol{X}_i \boldsymbol{X}_i^H) \boldsymbol{I}
\right)^{-1}}_{} \nonumber \\ & \times \bar{\boldsymbol{H}}_k
\boldsymbol{X}_k. \label{eq-ite-W}
\end{align}

\subsubsection{Update $\boldsymbol{\Lambda}_k$}
Given other variables fixed, the update of
$\boldsymbol{\Lambda}_k$ is given by
\begin{align}
\boldsymbol{\Lambda}_k^{} =( \boldsymbol{I} -
\boldsymbol{W}_k^H\bar{\boldsymbol{H}}_k\boldsymbol{X}_k)^{-1}. \label{eq-ite-Lam}
\end{align}

\subsubsection{Update $\boldsymbol{X}_k$}
Given other variables fixed, the optimization of
$\boldsymbol{X}_k$ is formulated as
\begin{align}
\min_{\boldsymbol{X}_k} \; &  \sum_{k=1}^K \left( \rho_c w_k {\rm tr}
\left( \boldsymbol{\Lambda}_k \boldsymbol{E}_{k}  \right) \right)  +
\rho_s \| \boldsymbol{u}\|_2^2   {\rm tr} ( \boldsymbol{X}^H {\tilde{\boldsymbol{B}}} \boldsymbol{X})  \nonumber \\
 & - 2\rho_s   \Re \{ \boldsymbol{w}^H \bar{\boldsymbol{A}}
 \boldsymbol{X} \boldsymbol{u} \}  + \frac{\rho_s \sigma^2}{P_t}
 \| \boldsymbol{u}\|_2^2 \| \boldsymbol{w}\|_2^2
 \sum_{i=1}^K {\rm tr}(\tilde{\boldsymbol{H}}\boldsymbol{X}_i \boldsymbol{X}_i^H),
\end{align}
where $\boldsymbol{E}_k$ is defined in \eqref{eq-Ek} and
$\tilde{\boldsymbol{B}}= \sum_{i=1}^I \bar{\boldsymbol{B}}_i^H
\boldsymbol{w} \boldsymbol{w}^H \bar{\boldsymbol{B}}_i$.
Similarly, by checking its first-order optimality condition, the update
of $\boldsymbol{X}_k$ is given by
\begin{align}
&\boldsymbol{X}_k \nonumber \\ =& \bigg(  \rho_c \sum_{i=1}^K w_i
\bar{\boldsymbol{H}}_i^H \boldsymbol{W}_i  \boldsymbol{\Lambda}_i
\boldsymbol{W}_i^H \bar{\boldsymbol{H}}_i + \rho_s \| \boldsymbol{u}\|_2^2 {\tilde{\boldsymbol{B}}} \nonumber \\
&  + \bigg( \rho_c  \sum_{i=1}^K\frac{ w_i \sigma_i^2}{P_t} {\rm tr}
( \boldsymbol{\Lambda}_i \boldsymbol{W}_i^H \boldsymbol{W}_i  ) +
\frac{\rho_s\sigma^2}{P_t} \| \boldsymbol{u}\|_2^2 \| \boldsymbol{w}\|_2^2
\bigg)\tilde{\boldsymbol{H}} \bigg)^{-1}  \nonumber \\
&\times \left( \rho_c w_k \bar{\boldsymbol{H}}_k^H
\boldsymbol{W}_k \boldsymbol{\Lambda}_k +
\rho_s\bar{\boldsymbol{A}}^H \boldsymbol{w} \boldsymbol{u}_k^H
\right). \label{eq-ite-X}
\end{align}

\subsubsection{Update $\boldsymbol{u}_k$}
Given other variables fixed, we can update $\boldsymbol{u}_k$ by
solving
\begin{align}
    \min_{\boldsymbol{u}_k} \quad  -
2\Re \{ \boldsymbol{w}^H \boldsymbol{A} \boldsymbol{X}
\boldsymbol{u} \} + \| \boldsymbol{u}\|_2^2  f_2( \boldsymbol{X},
\boldsymbol{w})
\end{align}
where $f_2( \boldsymbol{X},
\boldsymbol{w})$ is defined in \eqref{eq-opt-SCNR}. The optimal solution can be calculated as
\begin{align}
\boldsymbol{u}_k=  \frac{{\boldsymbol{X}_k^H}
\bar{\boldsymbol{A}}^H \boldsymbol{w}}{ \boldsymbol{w}^H
\mathring{\boldsymbol{B}}\boldsymbol{w} + \frac{\sigma^2}{P_t}
\sum_{i=1}^K {\rm tr}( \tilde{\boldsymbol{H}}  \boldsymbol{X}_i
\boldsymbol{X}_i^H) \boldsymbol{w}^H \boldsymbol{w}
},\label{eq-ite-u}
\end{align}
where
$\mathring{\boldsymbol{B}}=\sum_{i=1}^I(\bar{\boldsymbol{B}}_i
{\boldsymbol {X}} \boldsymbol {X}^H \bar{\boldsymbol {B}}_i^H)$.

Finally, let $(\{\boldsymbol{W}_k^{\star},
\boldsymbol{X}_k^{\star},\boldsymbol{\Lambda}_k^{\star}\}_{k=1}^K,
\boldsymbol{w}^{\star},\boldsymbol{u}^{\star})$ denote the
solution to problem \eqref{eq-opt-final}. Then the solution to
problem \eqref{eq-opt-simp0} is given as
$(\{\boldsymbol{W}_k^{\star}, \boldsymbol{F}_k^{\star} \}_{k=1}^K,
\boldsymbol{w}^{\star})$, where
\begin{align}
\boldsymbol{F}_k^{\star} = {\alpha}^{\star}
\boldsymbol{V}\boldsymbol{X}_k, \quad  \alpha^{\star} =
\sqrt{\frac{P_t}{\sum_{k=1}^K {\rm tr}(\tilde{\boldsymbol{H}}
\boldsymbol{X}_i^{\star} (\boldsymbol{X}_i^{\star})^H)}}.
\label{eq-opt-alpha}
\end{align}

For clarity, we summarize the proposed BCD-based method in
Algorithm \ref{Algorithm1}. Since the BCD generates a
non-decreasing sequence, the proposed algorithm is guaranteed to
converge to a stationary point of the optimization problem.

Now, we analyze the computational complexity of the proposed
BCD-based method. Calculating $\tilde{\boldsymbol{H}}$ in
\eqref{eq-opt-X} has complexity $\mathcal{O}(N_t r^2)$, which is
linear in $N_t$. By examining each iteration of the proposed
BCD-based method, we see that the complexity of one iteration is
dominated by the matrix inverse operation involved in
\eqref{eq-ite-w}, \eqref{eq-ite-W}, \eqref{eq-ite-Lam}, and
\eqref{eq-ite-X}. Among them, the dominant term lies in
\eqref{eq-ite-X}, where the matrix inverse operation involves a
computation complexity in the order of $\mathcal{O}(r^3)$, where
$r$, defined in Proposition \ref{proposition1}, represents the
total number of signal paths from the BS to terminals (including
UEs and radar receiver) as well as clutter patches. As a
comparison, the traditional BCD-based framework requires to
perform $N_t$-dimensional matrix inverse, resulting in a
computational complexity in the order of $\mathcal{O}(N_t^3)$.
Note that for mmWave/THz systems, the number of antennas at the
BS, i.e., $N_t$, could be up to hundreds or even thousands in
order to compensate for the severe path loss, whereas $r$ is
usually small due to sparse scattering characteristic of
mmWave/THz channels. As a result, the proposed BCD method can
achieve a remarkable computational complexity reduction as
compared to traditional BCD-based method, while achieving the same
performance.

\begin{algorithm}[H]   
\caption{Proposed BCD-based algorithm for digital transceiver design}
\label{Algorithm1}
\begin{algorithmic}[1]
\REQUIRE Any initial point $\boldsymbol{X}$ satisfying the power
constraint in \eqref{eq-opt-X}, $\boldsymbol{\Lambda}_k =
\boldsymbol{I}, \boldsymbol{W}_k = \boldsymbol{I}, \forall k$, and
$\boldsymbol{w}$. Set the tolerance of accuracy $\epsilon$;
\REPEAT \STATE Update the radar receive beamforming vector
$\boldsymbol{w}$ via \eqref{eq-ite-w}; \STATE update the combiner
$\boldsymbol{W}_k, \forall k$ via \eqref{eq-ite-W}; \STATE update
the auxiliary variable $\boldsymbol{\Lambda}_k, \forall k$ via
\eqref{eq-ite-Lam}; \STATE update the auxiliary precoder
$\boldsymbol{X}_k, \forall k$ via \eqref{eq-ite-X}; \STATE update
the auxiliary varible $\boldsymbol{u}_k, \forall k$ via
\eqref{eq-ite-u}; \UNTIL{ The absolute difference of objective
function value in two consecutive iterations are smaller than
$\epsilon$ ;} \STATE Normalize the digital precoder
$\boldsymbol{F}_k^{\star}, \forall k$ via \eqref{eq-opt-alpha}
\RETURN{The optimized point
 $(\{\boldsymbol{W}_k^{\star}, \boldsymbol{F}_k^{\star} \}_{k=1}^K, \boldsymbol{w}^{\star})$}
\end{algorithmic}
\end{algorithm}

\subsection{Hybrid Precoder/Combiner Design}
\label{sec-hybrid} In this section, we employ a classical least
square-based approach to find hybrid precoder/combiner to
approximate the fully-digital precoder/combiner. Here, we focus on
the precoder design. The idea can be straightforwardly extended to
the combiner design. A natural approach is to directly minimize
the Euclidean distance between the hybrid precoder and the optimal
digital precoder obtained in Section \ref{sec-prop-BCd}. The
problem can be formulated as
\begin{align}
\min_{\boldsymbol{F}_{\rm RF}, \{\boldsymbol{F}_{{\rm BB},k}\}_{k=1}^K}
\quad &\sum_{k=1}^K \| \boldsymbol{F}_k^{\star} -
\boldsymbol{F}_{{\rm RF}} \boldsymbol{F}_{{\rm BB},k} \|_F^2 \nonumber \\
{\text{s.t.}} \quad & | \boldsymbol{F}_{{\rm RF}}(i,j) |= 1 \nonumber \\
& \sum_{k=1}^K \| \boldsymbol{F}_{{\rm RF}} \boldsymbol{F}_{{\rm
BB},k} \|_F^2 \leq P_t.\label{eq-approx}
\end{align}

\subsubsection{Initialization}
A good initial point is essential for expediting the convergence
speed of the algorithm. For hybrid precoding, we have
$\boldsymbol{F}_k = \boldsymbol{F}_{{\rm RF}} \boldsymbol{F}_{{\rm
BB},k}, \forall k$. This formulation implies that all hybrid
precoders lie in the range space of $\boldsymbol{F}_{{\rm RF}}$,
i.e., $\mathcal{R}(\boldsymbol{F}_k) \subset
\mathcal{R}(\boldsymbol{F}_{{\rm RF}}),\forall k$. Denote
$\boldsymbol{F}^{\star} = [ \boldsymbol{F}^{\star}_1 \phantom{0}
\boldsymbol{F}^{\star}_2 \ldots  \boldsymbol{F}^{\star}_K] \in
\mathbb C^{N_t \times KN_s}  $ as the matrix comprising all
optimized digital precoders, where each $\boldsymbol{F}_k^{\star}$
is given in \eqref{eq-opt-alpha}. The analog precoder
$\boldsymbol{F}_{\rm RF}$ is supposed to satisfy
$\mathcal{R}(\boldsymbol{F}^{\star}) \subset
\mathcal{R}(\boldsymbol{F}_{\rm RF})$ to achieve optimal
performance. Nevertheless, such a condition may not be satisfied
due to the limited number of RF chains. To obtain a decent
approximation performance, we hope to choose
$\mathcal{R}(\boldsymbol{F}_{\rm RF} )$ to maximize the dimension
of $\mathcal{R}(\boldsymbol{F}^{\star})  \cap
\mathcal{R}(\boldsymbol{F}_{\rm RF})$. To this goal, denote the
ordered SVD of $\boldsymbol{F}^{\star} = \tilde{\boldsymbol{U}}
\tilde{\boldsymbol{\Sigma}} \tilde{\boldsymbol{V}}^H$, the initial
point of $(\boldsymbol{F}_{\rm RF}^{\rm ini}, \{
\boldsymbol{F}_{{\rm BB},k} \}_{k=1}^{K})$ is given by
\begin{align}
\boldsymbol{F}_{\rm RF}^{\rm ini} =
\exp(j\angle(\tilde{\boldsymbol{U}}(:,1:N_t^{\rm RF})),
\boldsymbol{F}_{{\rm BB},k} =c_k(\boldsymbol{F}_{\rm RF}^{\rm
ini})^{\dagger} \boldsymbol{F}_{k}^{\star}, \forall k,
\end{align}
where $\angle (\cdot)$ denotes the element-wise argument operator,
$\tilde{\boldsymbol{U}}(:,1:N_t^{\rm RF})$ denotes the submatrix
formed by extracting the first $N_t^{\rm RF}$ columns of
$\tilde{\boldsymbol{U}}$, $(\boldsymbol{F}_{\rm RF}^{\rm
ini})^{\dagger}$ is the Moore-Penrose pseudo-inverse of
$\boldsymbol{F}_{\rm RF}^{\rm ini} $, and $c_k$ is a normalization
factor accounting for the transmit power constraint.

\subsubsection{Manifold Optimization-Based Algorithm}
The problem in \eqref{eq-approx} can be solved by a fast manifold
optimization-based algorithm \cite{WangFang20a,Kasai18}, in which we optimize
$\boldsymbol{F}_{\rm RF}$ and $\boldsymbol{F}_{{\rm BB},k},
\forall k$ in an alternating manner. In particular, optimization
of $\boldsymbol{F}_{\rm RF}$ is solved on a complex circle
manifold which is the product of $N_tKN_s$ complex circles. The
reader may refer to \cite{Kasai18} for more details. It should be
mentioned that the algorithm in \cite{Kasai18} is guaranteed to converge
to a critical point with a computational complexity at the order
of $\mathcal{O}(N_tN_{t}^{\rm RF} KN_s)$.

\subsubsection{How Many RF Chains Are Required?}
As discussed earlier, the performance of employing a hybrid
precoder $\boldsymbol{F}_k = \boldsymbol{F}_{{\rm RF}}
\boldsymbol{F}_{{\rm BB},k}, \forall k$ is inherently constrained
by the number of RF chains available at the BS. A fundamental
question is: for the considered ISAC system, how many RF chains
are required in order to achieve the same performance attained by
the fully-digital precoder. We have the following result
concerning the above question.

\begin{Proposition}
For hybrid precoder ISAC systems, it is sufficient that the number of RF chains
satisfies
\begin{align}
N_{t}^{\rm RF} \geq r
\end{align}
in order to achieve a performance of a fully digital precoder.
Here $r$, defined in Proposition \ref{proposition1}, denotes the
total number of resolvable paths from the BS to terminals as well
as clutters. Specifically, if the clutter nulling is achieved by
radar receive beamforming, then this condition can be further
relaxed as $N_{t}^{\rm RF} \geq\tilde{r} = r- I$ .
\label{proposition-hybrid-suff}
\end{Proposition}
\begin{proof}
As shown in Proposition \ref{proposition1}, the optimized digital
precoder can be expressed as $\boldsymbol{F}_k^{\star} =
\boldsymbol{V} \boldsymbol{X}_k^{\star}, \forall k$ with
$\boldsymbol{V} \in \mathbb C^{N_t \times r}$ and
$\boldsymbol{X}_k^{\star} \in \mathbb C^{r \times N_s}, \forall
k$. Note that $\boldsymbol{V}$ is a matrix comprising steering
vectors from the BS to all $K$ UEs, the target and all clutters.
As a result, all elements in $\boldsymbol{V}$ satisfy the unit
modulus constraint. If $N_t^{\rm RF} \geq r$, then we simply set
the analog precoder $\boldsymbol{F}_{\rm RF}$ as
\begin{align}
\boldsymbol{F}_{\rm RF} = [\boldsymbol{V}^T \phantom{0}
\boldsymbol{U}^T]^T \in \mathbb C^{N_t \times N_t^{\rm RF}}, \label{eq-hyU}
\end{align}
where $\boldsymbol{U}\in \mathbb C^{N_t\times (N_t^{\rm RF}-r)}$
is an arbitrary matrix satisfying the unit modulus constraint.
With $\boldsymbol{F}_{\rm RF}$ defined above, the optimal precoder
$\boldsymbol{F}_k^{\star}$ can always be attained by the hybrid
precoder via searching for an appropriate $\boldsymbol{F}_{{\rm BB},k}$
for each user.

On the other hand, if the radar receiver beamforming
$\boldsymbol{w}$ can effectively null those clutter patches, i.e.,
$\|\boldsymbol{w}^H \sum_{i=1}^I \boldsymbol{B}_i \|_2^2 = 0$, the
SCNR term in \eqref{grad-Fk} simplifies to $
\nabla_{\boldsymbol{F}_k} \text{SCNR} =\frac{1 }{f_2}
\boldsymbol{A}^H\boldsymbol{w}\boldsymbol{w}^H \boldsymbol{A}
\boldsymbol{F}_k$. Consequently, the paths associated with these
nulled clutter patches become insignificant for SCNR optimization
and can be eliminated from $\boldsymbol{V}$. This leads to the
optimized digital precoder of the form $\boldsymbol{F}_k =
\boldsymbol{V}_r \boldsymbol{X}_r$ with $\boldsymbol{V}_r \in
\mathbb C^{N_t \times (r-I)}$ and $\boldsymbol{X}_r \in \mathbb
C^{ (r-I) \times N_s}$. We can then define the hybrid precoder as
before using $\boldsymbol{V}_r$ instead of $\boldsymbol{V}$. This
completes the proof.
\end{proof}

Since mmWave/THz channels typically exhibits a sparse scattering
characteristic, we usually have $r \ll N_t$. Hence a small number
of RF chains are adequate for fully exploiting the multiplexing
gains of mmWave/THz MU-MIMO ISAC systems.

\section{A Simple Sub-Optimal Solution for Transceiver Design}\label{sec-linear}
The algorithm proposed in Section \ref{sec-prop-BCd} involves an
iterative procedure for updates of block variables, thereby posing
challenges in deriving clear insights. In this section, we propose
a simple sub-optimal solution to the optimization problem
\eqref{eq-opt-simp0}. The proposed solution is inspired by the
traditional BD scheme. Note that the
precoder $\boldsymbol{F}_k, \forall k$ should be optimized to
maximize the communication as well as sensing performance.
Naturally we can express the precoder $\boldsymbol{F}_k$ as a sum
of two terms, one for communication and the other for
sensing, i.e.
\begin{align}
\boldsymbol{F}_k = {\boldsymbol{F}}_{c,k} + {\boldsymbol{F}}_{s,k}
,
\end{align}
where ${\boldsymbol{F}}_{c,k}$ denotes the precoder matrix used
for the $k$th UE, and ${\boldsymbol{F}}_{s,k}$ is the precoder
used for sensing. Let $p_{c,k} = \| {\boldsymbol{F}}_{c,k}
\|_2^2,\forall k$ denote the transmit power for communication and
$p_{s,k} = \| {\boldsymbol{F}}_{s,k}\|_2^2, \forall k$ the
transmit power for sensing.


\subsection{Communication-Oriented Precoder Design}
\label{sec-Comm-prec} Firstly, we introduce the basic idea of
designing the precoder ${\boldsymbol{F}}_{c,k}$. The design of
${\boldsymbol{F}}_{c,k}$ is primarily based on the following two
criteria:
\begin{enumerate}
\item The precoder ${\boldsymbol{F}}_{c,k}$ is designed to cause no interference to other
UEs. Also, it should not illuminate any clutter patches to
deteriorate the sensing performance.
\item The precoder is designed to maximize the achievable rate for the $k$th UE.
\end{enumerate}

The above two criteria are introduced to ensure that the
communication-oriented precoder ${\boldsymbol{F}}_{c,k}$ enhances
the WSR while without compromising the sensing performance.
Specifically, the first criterion leads to
\begin{align}
\boldsymbol{H}_j \boldsymbol{F}_{c,k} = \boldsymbol{0}, \forall j
\neq k, \quad \boldsymbol{B}_i\boldsymbol{F}_{c,k} =
\boldsymbol{0}, \forall i=1,2,\ldots, I.\label{eq-constr-interf}
\end{align}
To satisfy the above constraints, the precoder
$\boldsymbol{F}_{c,k}$ should lie in the null space of the channel
from the BS to all other UEs and all clutter patches, i.e.,
\begin{align}
\hat{\boldsymbol{H}}_k = & [\boldsymbol{H}_1^T\phantom{0}
\ldots\phantom{0}\boldsymbol{H}_{k-1}^T\phantom{0}
\boldsymbol{H}_{k+1}^T\phantom{0}\ldots\phantom{0}\boldsymbol{H}_K^T\phantom{0}
\boldsymbol{B}_{1}^T\phantom{0}
\ldots\phantom{0}  \boldsymbol{B}_{I}^T]^T, \nonumber \\
\stackrel{(a)}=& \hat{\boldsymbol{U}} \hat{\boldsymbol{\Sigma}} [
\hat{\boldsymbol{V}}_k^{(1)}\phantom{0} \hat{\boldsymbol{V}}_k^{(0)}]^H,
\end{align}
where $(a)$ denotes the SVD of $\hat{\boldsymbol{H}}_k $, and
$\hat{\boldsymbol{V}}_k^{(0)}$ is a submatrix consisting of $N_t -
{\rm rank}(\hat{\boldsymbol{H}}_k )$ right singular vectors which
form an orthogonal basis for the null space of
$\hat{\boldsymbol{H}}_k$. Therefore, the precoder
$\boldsymbol{F}_{c,k}$ can be represented as $\boldsymbol{F}_{c,k}
= \hat{\boldsymbol{V}}_k^{(0)} \hat{\boldsymbol{X}}_{c,k} $, where
$\hat{\boldsymbol{X}}_{c,k}$ is to be determined.

It is noted that the interference-nulling technique enables us to
express the WSR term in \eqref{eq-R} as a sum of $K$
independent rate functions. As a result, the problem of designing
${\boldsymbol{F}}_{c,k}$ can be simplified as
\begin{align}
\max_{\boldsymbol{X}_{c,k}} \quad &  \log_2 {\rm det}(\boldsymbol{I} +
\sigma_k^{-2} (\boldsymbol{W}_k^H\boldsymbol{W}_k)^{-1}\boldsymbol{W}_k^H
\boldsymbol{H}_k^{\rm BD} \hat{ \boldsymbol{X}}_{c,k} \nonumber \\
&\times \hat{ \boldsymbol{X}}_{c,k}^H (\boldsymbol{H}_k^{\rm BD})^H \boldsymbol{W}_k )  \nonumber \\
{\text{s.t.}} \quad & {\rm tr}(\boldsymbol{X}_{c,k}
\boldsymbol{X}_{c,k}^H) \leq p_{c,k},
\label{eq-opt-Xck}
\end{align}
where $\boldsymbol{H}_k^{\rm BD}\triangleq
\boldsymbol{H}_k\hat{\boldsymbol{V}}_k^{(0)}$ denotes the
effective channel between the BS and the $k$th UE, and its ordered
SVD is given as
\begin{align}
\boldsymbol{H}_k^{\rm BD} =\boldsymbol{U}_{{\rm BD},k}
\boldsymbol{\Sigma}_{{\rm BD},k} \boldsymbol{V}_{{\rm BD},k}^H,
\label{eq-Hkbd}
\end{align}
where $\boldsymbol{\Sigma}_{{\rm BD},k}={\rm diag}(\varrho_{k,i},
\ldots, \varrho_{k,N_s})$. According to the traditional MIMO
theory\cite{tse2005fundamentals}, the optimal solution to \eqref{eq-opt-Xck} can be
obtained as
\begin{align}
\boldsymbol{X}_{c,k} =& \boldsymbol{V}_{{\rm BD},k}(:,1:N_s) \boldsymbol{\Gamma}_k^{1/2},  \\
    \boldsymbol{W}_k =& \boldsymbol{U}_{{\rm BD},k}(:,1:N_s), \label{eq-FBD}
\end{align}
where $\boldsymbol{V}_{\rm BD}(:,1:N_s)$ denotes the matrix
consisting of the first $N_s$ columns of $\boldsymbol{V}_{\rm
BD}$, $\boldsymbol{\Gamma}_k = {\rm diag}(p_{k,1}, \ldots,
p_{k,N_s})$ is a diagonal matrix with the $i$th diagonal element
$p_{k,i}$ being the power allocated to the corresponding data
stream. Also, we have $\sum_{i=1}^{N_s} p_{k,i} = p_{c,k}$. The
corresponding precoder $\boldsymbol{F}_{c,k}$ is thus given by
\begin{align}
\boldsymbol{F}_{c,k} =&\hat{\boldsymbol{V}}_k^{(0)}
\boldsymbol{V}_{{\rm BD},k}(:,1:N_s) \boldsymbol{\Gamma}_k^{1/2},
\label{eq-Fck}
\end{align}
We will discuss the power allocation $p_{k,i}, \forall k, \forall
i$ in the next subsection.


\subsection{Sensing-Oriented Precoder Design}\label{sec-sens-prec}
We express the sensing-oriented precoder as
\begin{align}
\boldsymbol{F}_{s,k} = \sqrt{p_{s,k}}\boldsymbol{F}_{\rm sen} ,
\quad \|\boldsymbol{F}_{\rm sen}\|_F^2 = 1, \label{eq-Fksen}
\end{align}
where $p_{s,k}$ denotes the power allocated to the sensing
precoder $\boldsymbol{F}_{\rm sen}$. The design of
$\boldsymbol{F}_{\rm sen}$ is based on the following two criteria
\begin{enumerate}
\item The precoder does not cause any interference to UEs. Meanwhile, it
should not illuminate any clutter patches to deteriorate the
sensing performance.
\item The precoder should maximize the received signal power at the target.
\end{enumerate}

Since the design criteria prevents the transmitter from illujminating the clutters, the radar receiver $\boldsymbol{w}$ reduces to a simple beamformer that points to the target, i.e., $\boldsymbol{w}=\boldsymbol{a}_s(\varphi_0)$. Consequently, the
target's received signal power due to the sensing precoder
$\boldsymbol{F}_{\rm sen}$ is given by ${\rm tr} (\boldsymbol{w}^H
\boldsymbol{A} \boldsymbol{F}_{\rm sen} \boldsymbol{F}_{\rm sen}^H
\boldsymbol{A} \boldsymbol{w}) = \sigma_A^2\|
\boldsymbol{a}_t^H(\theta_0) \boldsymbol{F}_{\rm sen} \|_2^2$.
Hence the above two criteria can be formulated into the following
optimization
\begin{align}
\max_{\boldsymbol{F}_{\rm sen} } \quad & \| \boldsymbol{a}_t^H(\theta_0)
\boldsymbol{F}_{\rm sen} \|_2^2 \nonumber \\
\text{s.t.} \quad & \bar{\boldsymbol{H}} \boldsymbol{F}_{\rm sen} = \boldsymbol{0}, \nonumber \\
& \|\boldsymbol{F}_{\rm sen}\|_F^2 = 1,\label{eq-opt-fsen}
\end{align}
where $\bar{\boldsymbol{H}} = [\boldsymbol{H}_1^T \phantom{0}
\ldots \phantom{0}  \boldsymbol{H}_K^T \phantom{0} \boldsymbol{B}_{1}^T \ldots
\boldsymbol{B}_{I}^T]^T $ denotes the matrix containing all UEs'
channels and clutter response matrices. The optimal solution to
problem \eqref{eq-opt-fsen} can be easily obtained by choosing the
columns in $\boldsymbol{F}_{\rm sen}$ as the orthogonal projection
of $\boldsymbol{a}_t(\theta_0)$ onto the null space of
$\bar{\boldsymbol{H}}$, i.e.,
\begin{align}
\boldsymbol{F}_{\rm sen} =& \boldsymbol{f}_{\rm sen} \times \boldsymbol{1}_{N_s}^T, \nonumber \\
\boldsymbol{f}_{\rm sen} =& c_{\rm sen}(\boldsymbol{a}_t(\theta_0) -
\bar{\boldsymbol{V}}\bar{\boldsymbol{V}}^H \boldsymbol{a}_t(\theta_0) ) ,\label{eq-fsen}
\end{align}
where $\bar{\boldsymbol{V}}$ is the right singular matrix of the
truncated SVD of $\bar{\boldsymbol{H}} =
\bar{\boldsymbol{U}}\bar{\boldsymbol{\Sigma}}
\bar{\boldsymbol{V}}^H$, and $c_{\rm sen}$ is a scalar to ensure
that $\|\boldsymbol{F}_{\rm sen} \|_F=1$.

\emph{Remarks}: When $\boldsymbol{a}_t(\theta_0)$ lies in the
range space of $\bar{\boldsymbol{H}}^H$,
$\boldsymbol{a}_t(\theta_0)$ is orthogonal to the null space of
$\bar{\boldsymbol{H}}$, resulting in $\boldsymbol{f}_{\rm sen} =
\boldsymbol{0}$. This is consistent with the intuition that when
the target subspace overlaps with the communication subspace, the
communication power itself suffices for target illumination.
Consequently, no additional power needs to be allocated for
sensing.



\subsection{Combining Communication- and Sensing-Oriented Precoder}
\label{sec-synth} To obtain the final precoder, it remains to
determine the power coefficients that are allocated to
$\boldsymbol{F}_{c,k}$ and $\boldsymbol{F}_{s,k}$.

Define $\mathcal{P} \triangleq \{ \{
p_{k,i}\}_{k=1,i=1}^{K,N_s},\; \{ p_{s,k} \}_{k=1}^K \}$ as the
set containing all power coefficients. Substituting the combiner
$\boldsymbol{W}_k$ in \eqref{eq-FBD} and the precoder in
\eqref{eq-Fksen} into the WSR term in \eqref{eq-R}, we arrive at
\begin{align}
R(\mathcal{P}) =& \sum_{k=1}^K  w_k \log_2 {\rm det}(\boldsymbol{I} +
\sigma_k^{-2} (\boldsymbol{W}_k^H\boldsymbol{W}_k)^{-1}\boldsymbol{W}_k^H
\boldsymbol{H}_k{ \boldsymbol{F}}_k \nonumber \\
&\times { \boldsymbol{F}}_k^H (\boldsymbol{H}_k)^H \boldsymbol{W}_k ) \nonumber \\
=& \sum_{k=1}^K w_k \sum_{i=1}^{N_s} \log_2 (1 +
\sigma_k^{-2}\varrho_{k,i}^2 p_{k,i}),\label{eq-wsr-linear}
\end{align}
where $\varrho_{k,i}$ is defined in \eqref{eq-Hkbd}. Fixing
$\boldsymbol{w} = \boldsymbol{a}_{s}(\phi_0)$, the SCNR term in
\eqref{eq-SCNR} can be expressed as
\begin{align}
\text{SCNR} =& \frac{\boldsymbol{w}^H  \boldsymbol {A}
\boldsymbol{F} \boldsymbol {F}^H  \boldsymbol{A}^H \boldsymbol{w}}{ \boldsymbol{w}^H
\left(\sum_{i=1}^I \boldsymbol{B}_i \boldsymbol {F} \boldsymbol {F}^H
\boldsymbol {B}_i^H + \sigma^2 \boldsymbol{I}   \right) \boldsymbol{w}  }  \nonumber \\
= & \frac{1}{ \sigma^2} \sum_{k=1}^K {\rm tr} ( \boldsymbol{F}_k^H
\boldsymbol{A}^H \boldsymbol{w} \boldsymbol{w}^H \boldsymbol{A} \boldsymbol{F}_k) \nonumber \\
=&  \sum_{k=1}^K \sum_{i=1}^{N_s} d_{k,i} p_{k,i} + g\sum_{k=1}^K p_{s,k} \nonumber \\
&+ \sum_{k=1}^K \sum_{i=1}^{N_s} e_{k,i} \sqrt{p_{k,i} p_{s,k}}
\nonumber \\
\stackrel{(a)}\approx &  \sum_{k=1}^K
\sum_{i=1}^{N_s} d_{k,i} p_{k,i} + g\sum_{k=1}^K p_{s,k},
\label{eq-SCNR_m}
\end{align}
where
\begin{align}
d_{k,i} \triangleq& \sigma^{-2} \sigma_A^2
[\boldsymbol{F}_{c,k}^H\boldsymbol{a}_t(\theta_0)
\boldsymbol{a}_t^H(\theta_0) \boldsymbol{F}_{c,k}]_{i,i}  , \label{eq-dki}\\
e_{k,i} \triangleq& 2\sigma^{-2} \sigma_A^2\Re \{
[\boldsymbol{F}_{\rm sen}^H
\boldsymbol{a}_t(\theta_0) \boldsymbol{a}_t^H(\theta_0)\boldsymbol{F}_{c,k}]_{i,i} \}, \\
g \triangleq& \sigma^{-2} \sigma_A^2{\rm
tr}\left(\boldsymbol{F}_{\rm sen}^H\boldsymbol{a}_t(\theta_0)
\boldsymbol{a}_t^H(\theta_0) \boldsymbol{F}_{{\rm sen}} \right),
\label{eq-g}
\end{align}
with $[\boldsymbol{R}]_{i,j}$ denoting the $(i,j)$th element of
the matrix $\boldsymbol{R}$. By examining \eqref{eq-dki} to
\eqref{eq-g}, we see that $d_{k,i}$ characterizes the SCNR gain
resulting from the energy leakage from communication-oriented
precoder, $g$ only depends on $\boldsymbol{F}_{\rm sen}$, and
$e_{k,i}$ is a cross term that describes the correlation between
the communication-oriented precoder $\boldsymbol{F}_{c,k}$ and the
sensing-oriented precoder $\boldsymbol{F}_{\rm sen}$. It should be noted that $e_{k,i}$ is usually much smaller than $d_{k,i}$ or $g$ because of the following two reasons:
\begin{itemize}
	\item when $\boldsymbol{a}_t(\theta_0)$ lies in the communication subspace, we have $\boldsymbol{F}_{\rm sen} = \boldsymbol{0}$ and thus $e_{k,i} = 0$;
	\item when $\boldsymbol{a}_t(\theta_0)$ does not lie in the communication subspace, the sensing-oriented precoder $\boldsymbol{F}_{\rm sen}$ lies in the null space of $\boldsymbol{H}_k$ while the communication-oriented precoder $\boldsymbol{F}_{c,k}$ should be as far away as possible from the null space of $\boldsymbol{H}_k$ to ensure communication performance. Specifically, when $\boldsymbol{f}_{\rm sen}^H \boldsymbol{F}_{c,k} = \boldsymbol{0}$, we have $e_{k,i} = 0$.
\end{itemize}


Since the cross term $e_{k,i}\sqrt{p_{k,i}p_{s,k}}$ makes it
difficult to optimize, we turn to maximize the sum of the WSR and
the approximate expression of SCNR derived in \eqref{eq-SCNR_m}. The problem
can be formulated as a simple power allocation problem
\begin{align}
\max_{\mathcal{P}} \quad & \rho_c \sum_{k=1}^{K}w_k \sum_{i=1}^{N_s} \log_2(1 +
\sigma_k^{-2}\varrho_{k,i}^2 p_{k,i}) \nonumber \\
&+ \rho_s  \sum_{k=1}^K\sum_{i=1}^{N_s} d_{k,i} p_{k,i}   + \rho_s g\sum_{k=1}^K  p_{s,k} \nonumber \\
{\text{s.t.}} \quad & \sum_{k=1}^{K} \sum_{i=1}^{N_s} p_{k,i} +
\sum_{k=1}^{K} p_{s,k} \leq P_t, \label{eq-pk}
\end{align}
which is convex and admits a closed-form solution akin to
traditional water-filling scheme.

\begin{Proposition}
The optimal solution of \eqref{eq-pk} is given as
\label{proposition-power}
\begin{align}
p_{k,i} =&\max \{ 0, \frac{\rho_c w_k}{(\ln 2)(\mu - \rho_s d_{k,i})} -
\frac{\sigma_k^2}{\varrho_{k,i}^2}\}, \mu \geq \rho_s g,\\
p_{s,k} =& \begin{cases} 0,& \text{if } \mu > \rho_s g, \\
\frac{P_t}{K} - \frac{1}{K} (\sum_{k=1}^K \sum_{i=1}^I p_{k,i}), &
\text{if } \mu = \rho_s g,
 \end{cases}
    \end{align}
where $\mu$ is the auxiliary variable satisfying the total
transmit power constraint in \eqref{eq-pk}.
\end{Proposition}
\begin{proof}
The solution can be obtained directly by solving the KKT condition
of the problem \eqref{eq-pk} and is thus omitted for brevity.
\end{proof}

Here we gain insights into the proposed solution. Firstly, it is
observed that, when either $\rho_s = 0$ or $g = 0$, the power
allocated to sensing becomes zero and the power allocated to
different data streams for different UEs admits the optimal
water-filling solution. The result for $\rho_s = 0$ is intuitive
since all the power should be allocated to improve the
communication performance in such a case. On the other hand, in
cases where $\rho_s>0$, the power allocated to sensing may still
be zero. This occurs if $\boldsymbol{a}_t(\theta_0)$ lies in the
range space of $\bar{\boldsymbol{H}}^H$, leading to $g=0$, as
discussed earlier in Section \ref{sec-sens-prec}.

Next, we delve into the scenario where $p_{s,k}>0$, i.e., $\mu =
\rho_s g$. It is interesting to observe that the power allocated
to different streams of communication users still admits a
solution akin to water-filling scheme. However, instead using a
fixed water level as in traditional MIMO capacity optimization
problems, the water level in \eqref{eq-pk}, i.e., $\frac{\rho_c
w_k}{ (\ln 2)(\rho_s g - \rho_s d_{k,i})}$, is dependent on
$d_{k,i}$ which can be considered as a metric characterizing the
amount of power leakage from communication to sensing. When
$d_{k,i}$ gets larger, the water level rises, resulting in an
increased power allocated to the $k$th communication UE. This
result is quite natural as increasing the power to this UE can
also help increase the sensing performance. In contrast, when
$d_{k,i}, \forall k,i$ is relatively small, i.e., UEs and target
are separated in spatial domain, and the communication users have
a high receive SNR, Proposition \ref{proposition-power} suggests
that the equal power allocation among different UEs is an
approximately optimal solution.

\emph{Remarks}: The proposed method for the problem \eqref{eq-pk}
can be easily adapted to the case where the objective is to
maximize the sensing (or communication) performance subject to a
communication (or sensing) constraint. In particular, we can use a
line search-based method to find a suitable value of $(\rho_c,
\rho_s)$ to satisfy the communication (or sensing) constraint and
meanwhile maximize the sensing (or communication) performance.
After obtaining the optimal digital precoder/combiner, we can
resort to the manifold-optimization method in Section
\ref{sec-hybrid} to obtain a hybrid precoder/combiner.

\begin{figure}[t]
\centering
\includegraphics[width=3.2in] {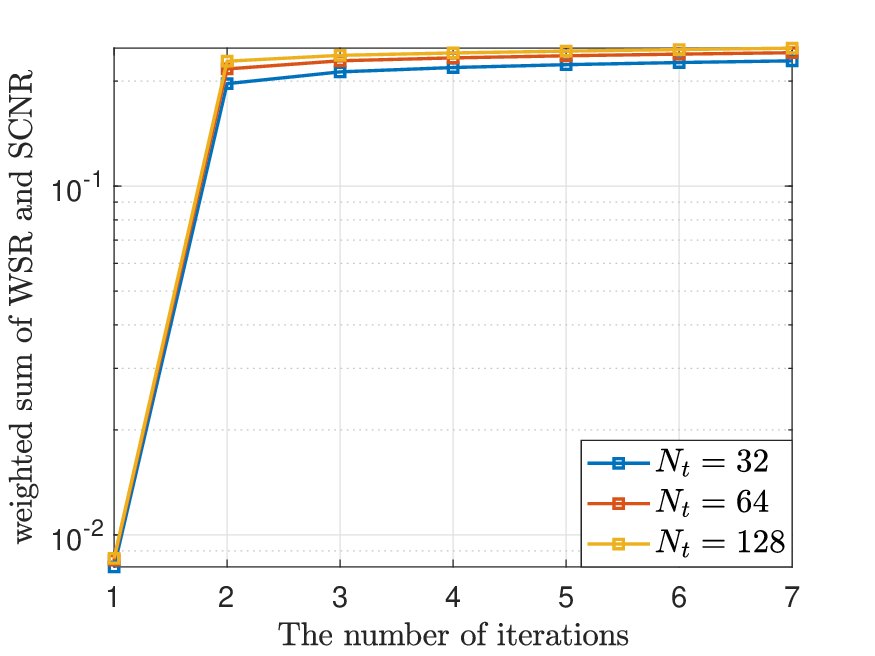}
\caption{Convergence behaviour of the proposed BCD-based method.}
\label{fig_converge}
\end{figure}

\begin{figure*}[!t]
 \centering
 \subfigure[Weighted sum of WSR and SCNR v.s. the transmit power $P_t$.]{
  \begin{minipage}{.31\textwidth}
   \centering
   \includegraphics[scale=.26]{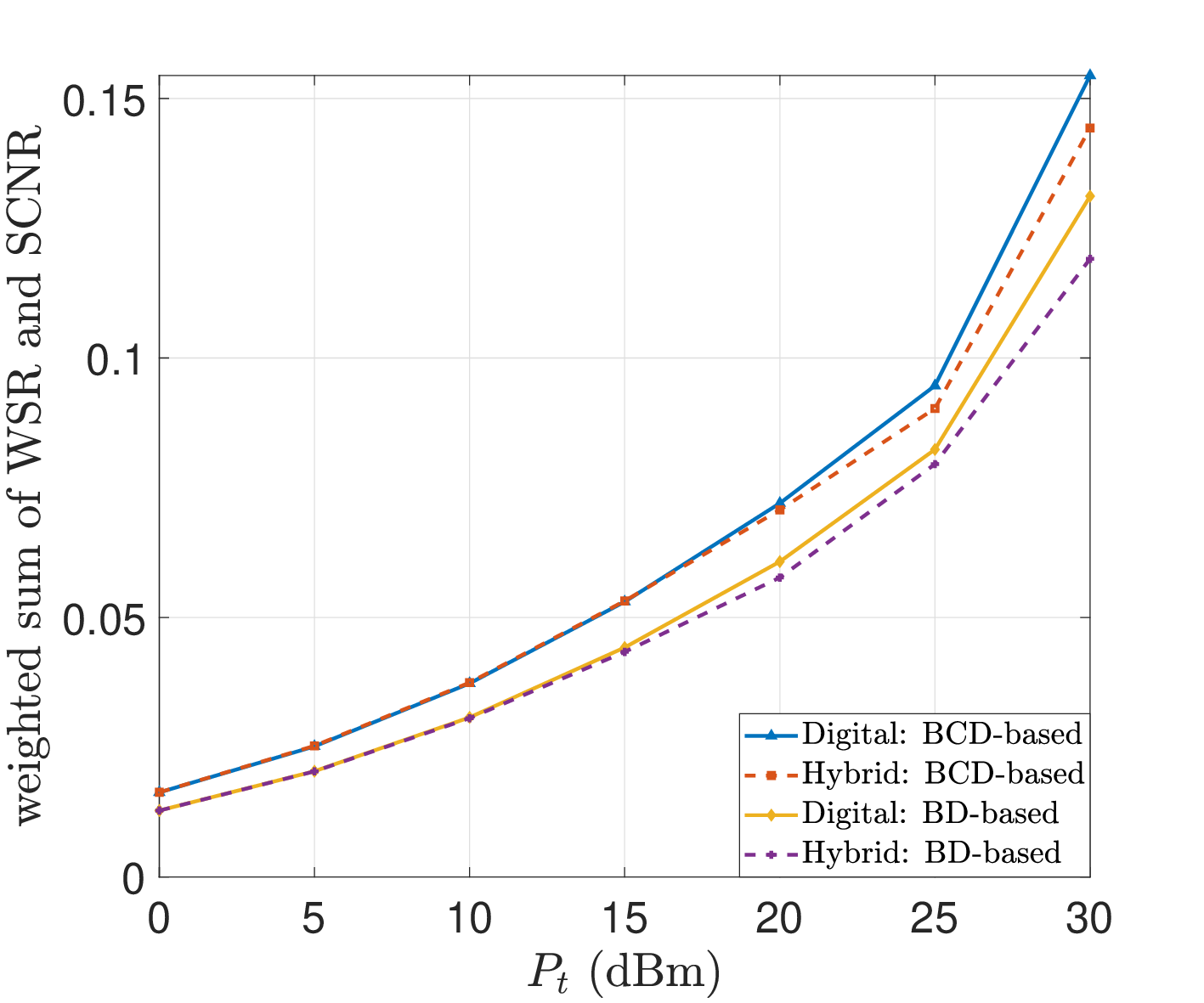}
  \end{minipage}
  \label{FIG6_1}
 }
 \subfigure[WSR v.s. the transmit power $P_t$.]{
  \begin{minipage}{.31\textwidth}
   \centering
   \includegraphics[scale=.26]{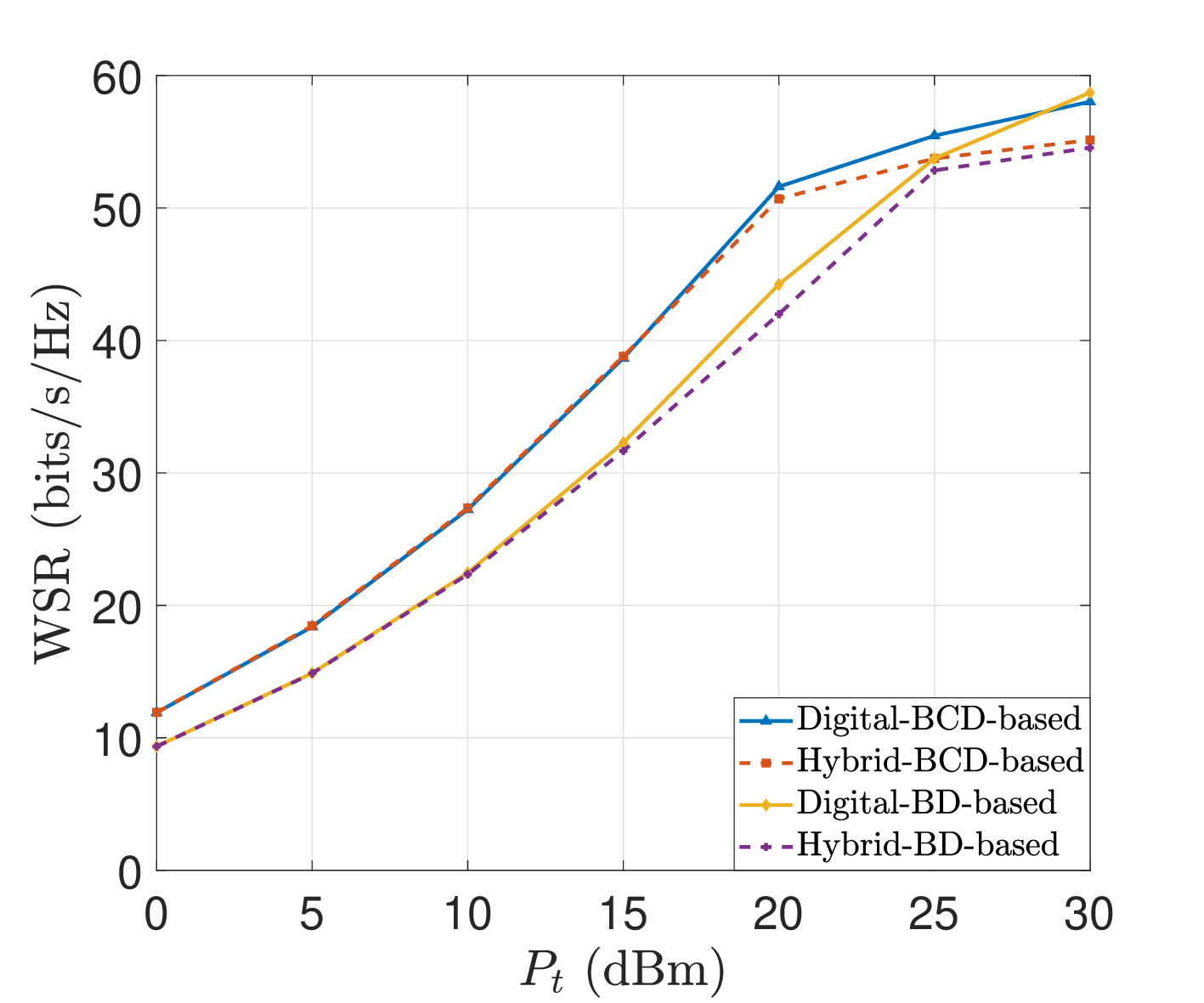}
  \end{minipage}
  \label{FIG6_2}
 }
 \subfigure[SCNR v.s. the transmit power $P_t$.]{
  \begin{minipage}{.31\textwidth}
   \centering
   \includegraphics[scale=.26]{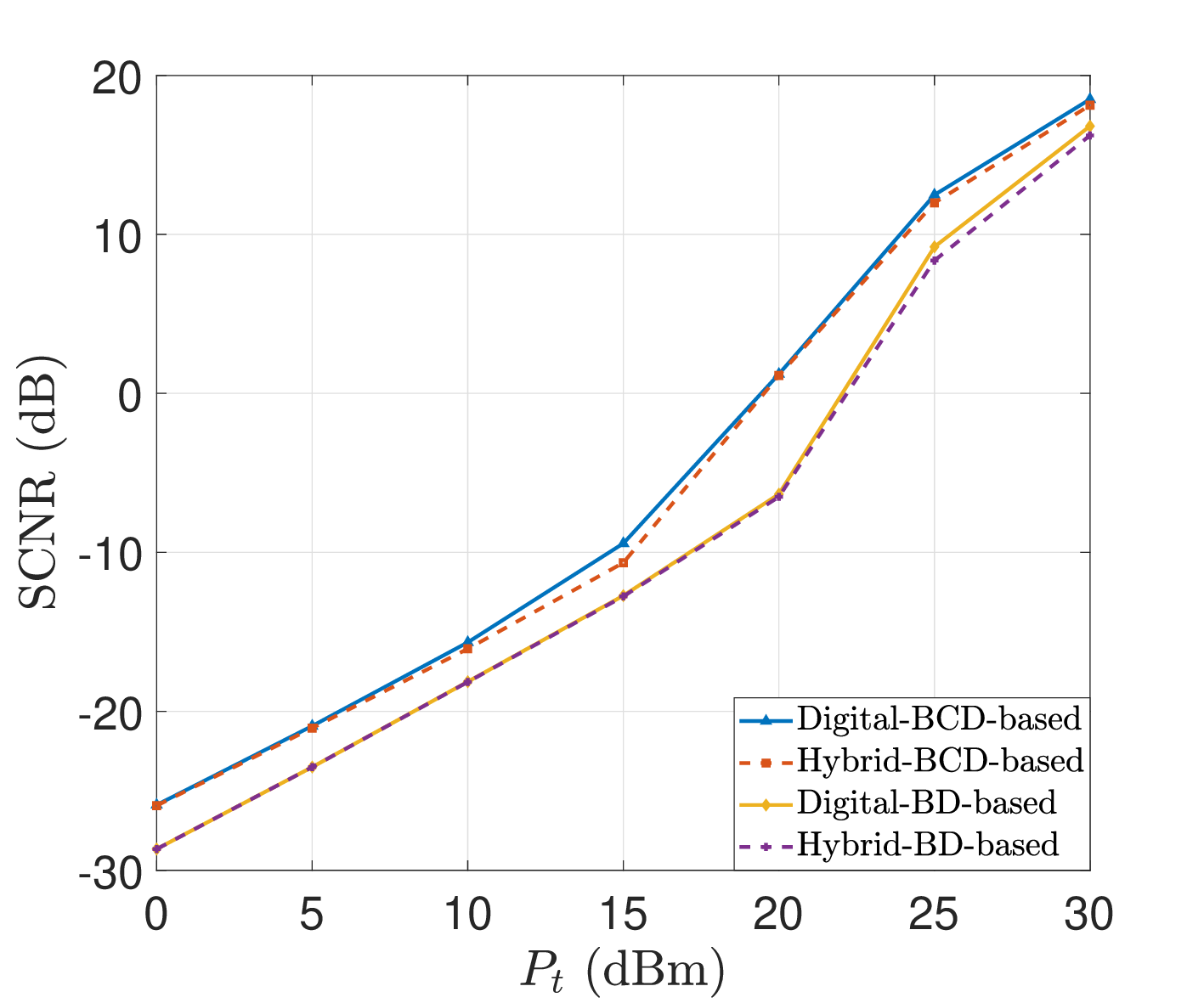}
  \end{minipage}
  \label{FIG6_3}
 }
  \caption{Weighted sum of WSR and SCNR, WSR, and SCNR v.s. the transmit power $P_t$.}
 \label{fig_Pt}
\end{figure*}

\begin{figure}[t]
\centering
\includegraphics[width=3.2in] {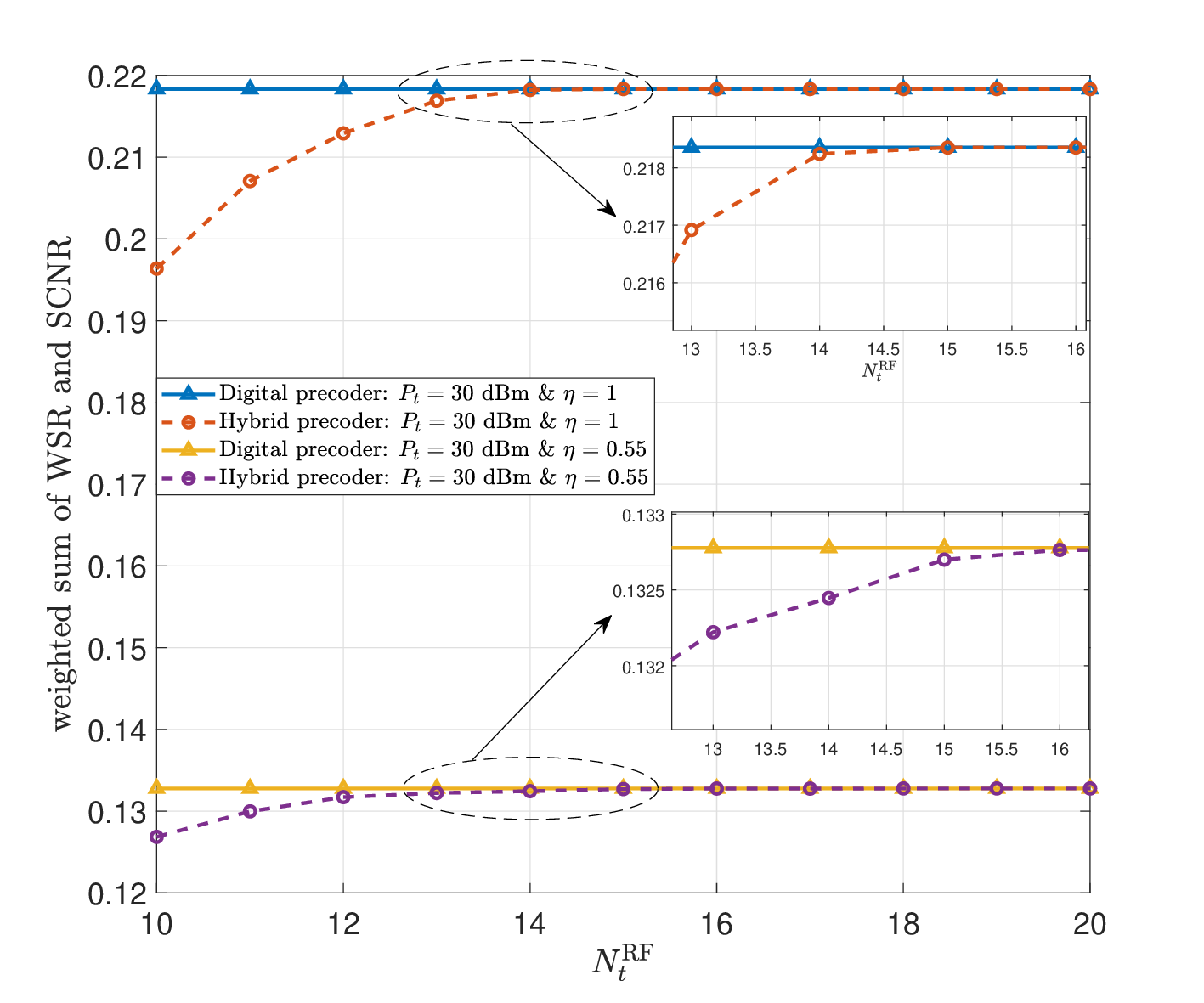}
\caption{Weighted sum of WSR and SCNR v.s. $N_t^{\rm RF}$.}
\label{fig_NtRF}
\end{figure}

\begin{figure}[t]
\centering
\includegraphics[width=3.2in] {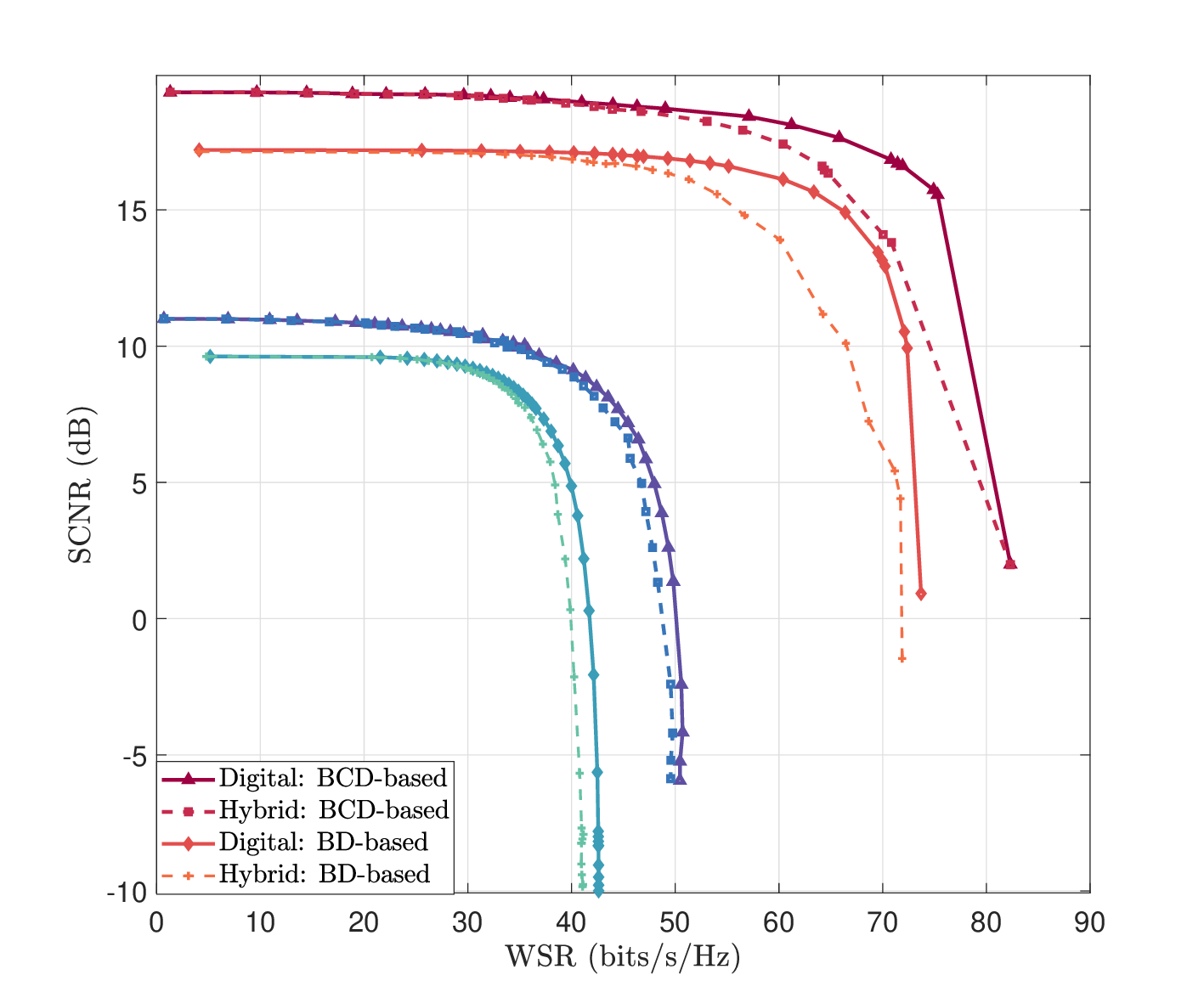}
\caption{SCNR-WSR performance region by varying $\eta$ in
$[0,1]$.} \label{fig_region}
\end{figure}

\section{Simulation Results}\label{sec-simu}
In this section, we evaluate the performance of the proposed
BCD-based method and the BD-based sub-optimal solution. Unless
otherwise stated, the simulation parameters are set as follows.
The BS, which is equipped with a uniform linear array (ULA) of
$N_t=64$ antennas, serves $K=5$ UEs. Each UE is equipped with a
ULA of $N_r=4$ antennas. The number of data streams is set to
$N_s=2$. The number of RF chains at the BS and the UE is set to
$N_{t}^{\rm RF} =KN_s = 5 \times 2 = 10$ and $N_r^{\rm RF} = N_s =
2$. The BS and the radar receiver are located at coordinates
$[20,30]^T$ and $[15,15]^T$, respectively. The angular parameters
associated with the UEs and the clutter patches are uniformly
generated from $[-\pi/2, \pi/2]$. The carrier frequency is set to
$28$ GHz. The noise power at each UE and the radar receiver is set
to $\sigma^2 = -90$ dBm.

For communication channels, the complex path gain of the LOS path
follows $\beta_{l_k} \sim \mathcal{CN}(0,10^{-0.1\kappa})$, where
$\kappa = a + 10b\log_{10} (d) + \epsilon$, with $d$ being the
distance between the BS and the $k$th UE and $\epsilon \in
\mathcal{CN}(0,\sigma_{\epsilon}^2)$ \cite{SunRappaort18}. Here we
set $a=61.4$, $b=2$, $\sigma_{\epsilon} = 5.8$ dB as suggested in
\cite{SunRappaort18}. For non-line-of-sight (NLOS) paths, the
complex gain follows a complex Gaussian distribution $\beta_{k,i}
\sim \mathcal{CN}(0,10^{-0.1(\kappa + \mu)})$ with $\mu=7$ dB
denoting the Rician factor \cite{xie2022Perceptive}. For sensing
response matrices, we set $\sigma_A^2/\sigma^2 = 20$ dB and
$\frac{1}{I} \sum_{i=1}^I \sigma_{B_i}^2/ \sigma^2 = 40$ dB for
the target and clutter patches, respectively. The weighting
coefficient for each UE is set to $w_k=1$.


\subsection{System Performance}
Fig. \ref{fig_converge} depicts the convergence behaviour of our
proposed BCD-based method, where we set the transmit power $P_t =
30$ dBm and the weighting coefficient $\eta=0.5$. It can be seen
that the proposed BCD-based method has a relatively fast
convergence rate and attains the maximum objective function value
within only $7$ iterations. Moreover, we observe that the
convergence rate is independent of the number of antennas at the
BS, i.e., $N_t$, which is a desirable characteristic for
large-scale mmWave systems.



\begin{figure*}[t]
 \centering
 \subfigure[Beam pattern for the first UE.]{
  \begin{minipage}{.45\textwidth}
   \centering
   \includegraphics[scale=.34]{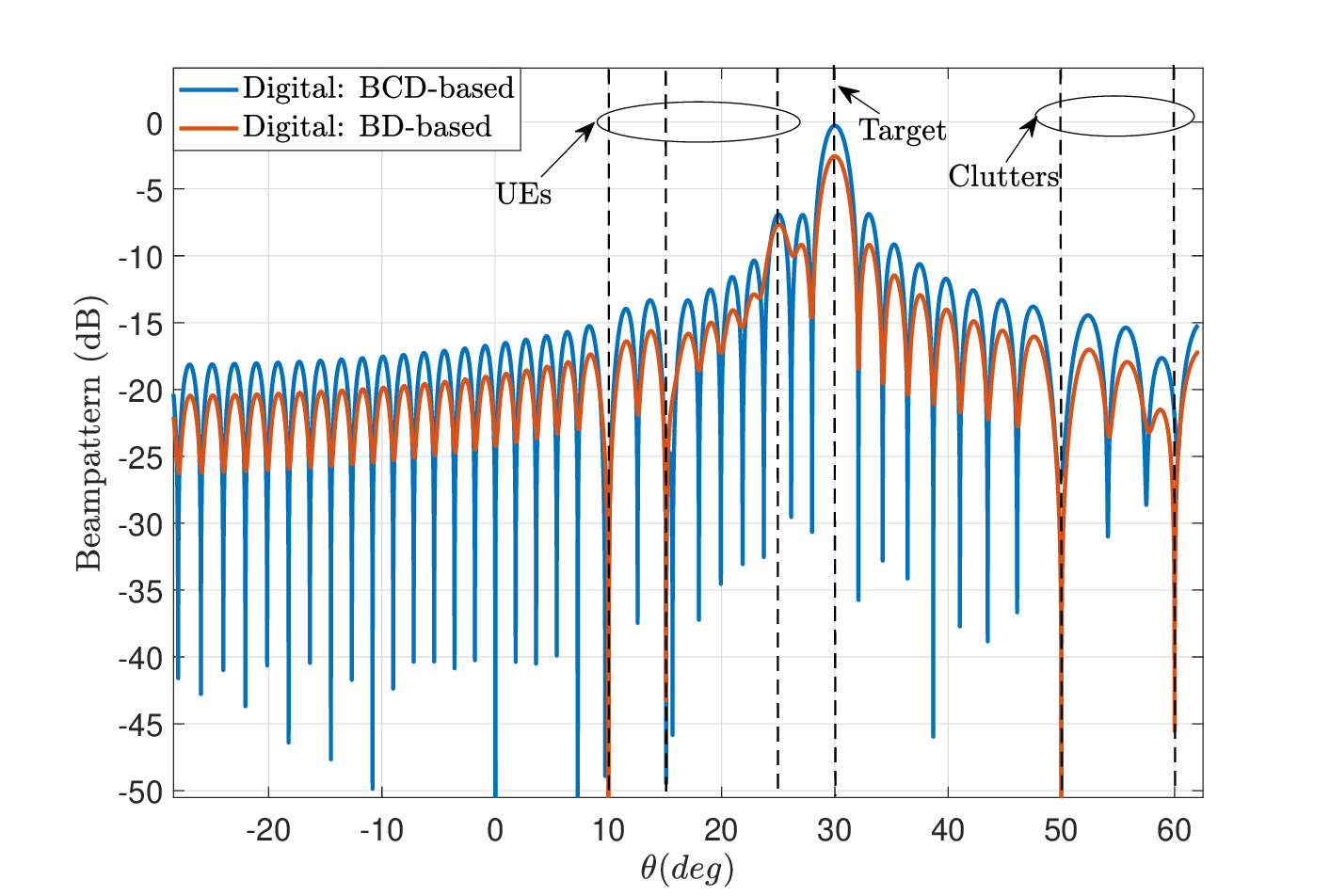}
  \end{minipage}
  \label{FIG7_1}
 }
 \subfigure[Beam pattern for the second UE.]{
  \begin{minipage}{.45\textwidth}
   \centering
   \includegraphics[scale=.34]{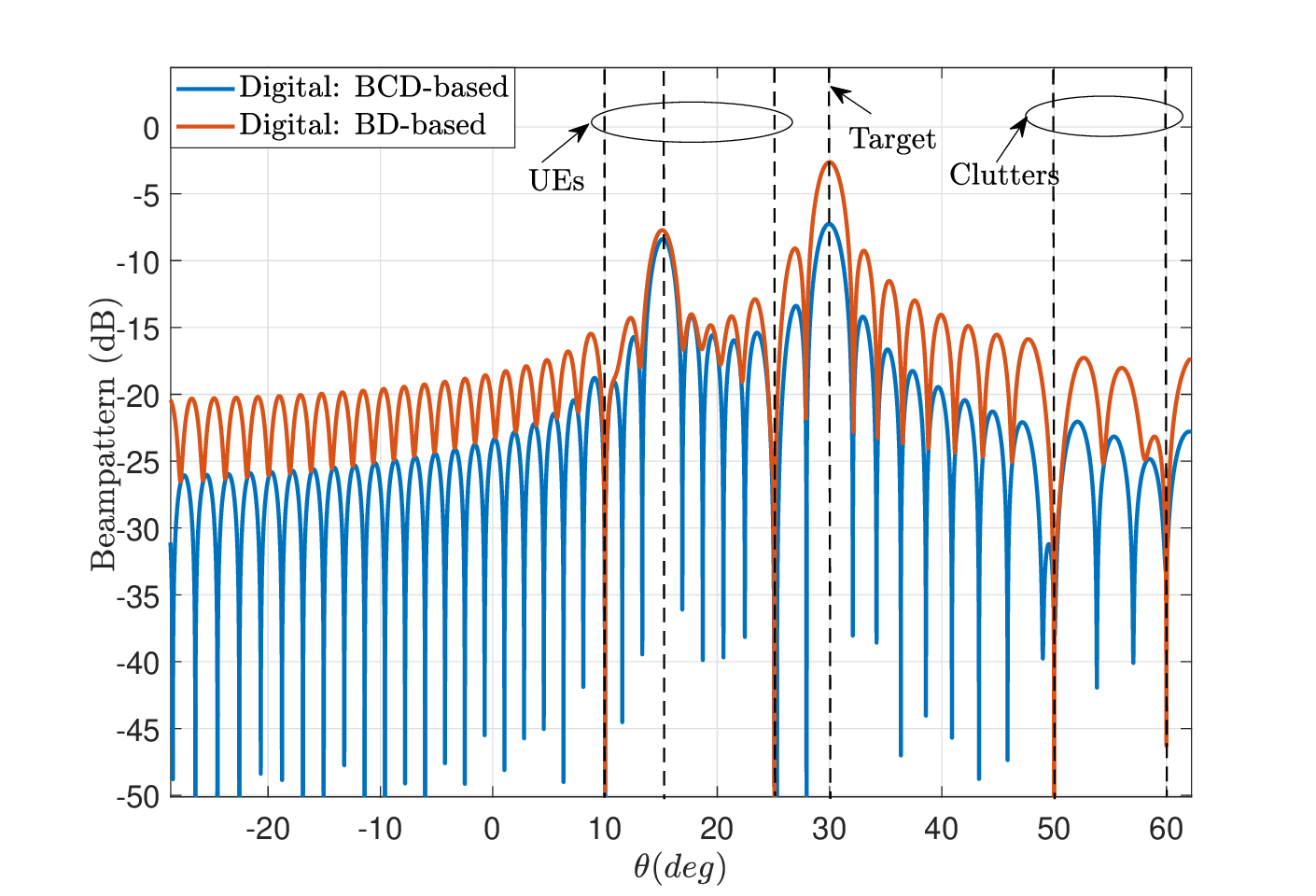}
  \end{minipage}
  \label{FIG7_2}
 }
 \hfill
 \subfigure[Beam pattern for the third UE.]{
  \begin{minipage}{.45\textwidth}
   \centering
   \includegraphics[scale=.34]{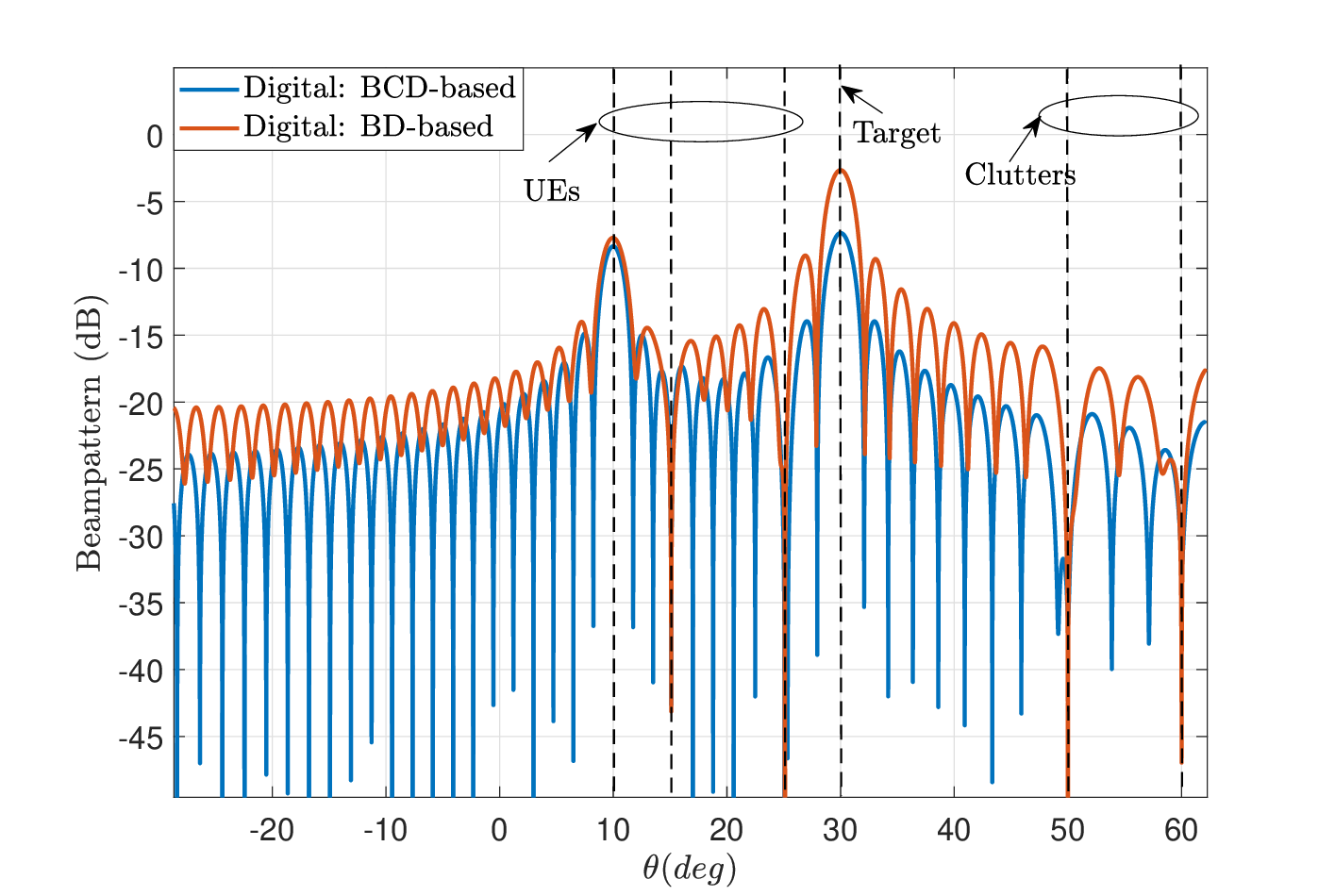}
  \end{minipage}
  \label{FIG7_3}
 }
 \subfigure[Overall beam pattern.]{
  \begin{minipage}{.45\textwidth}
   \centering
   \includegraphics[scale=.34]{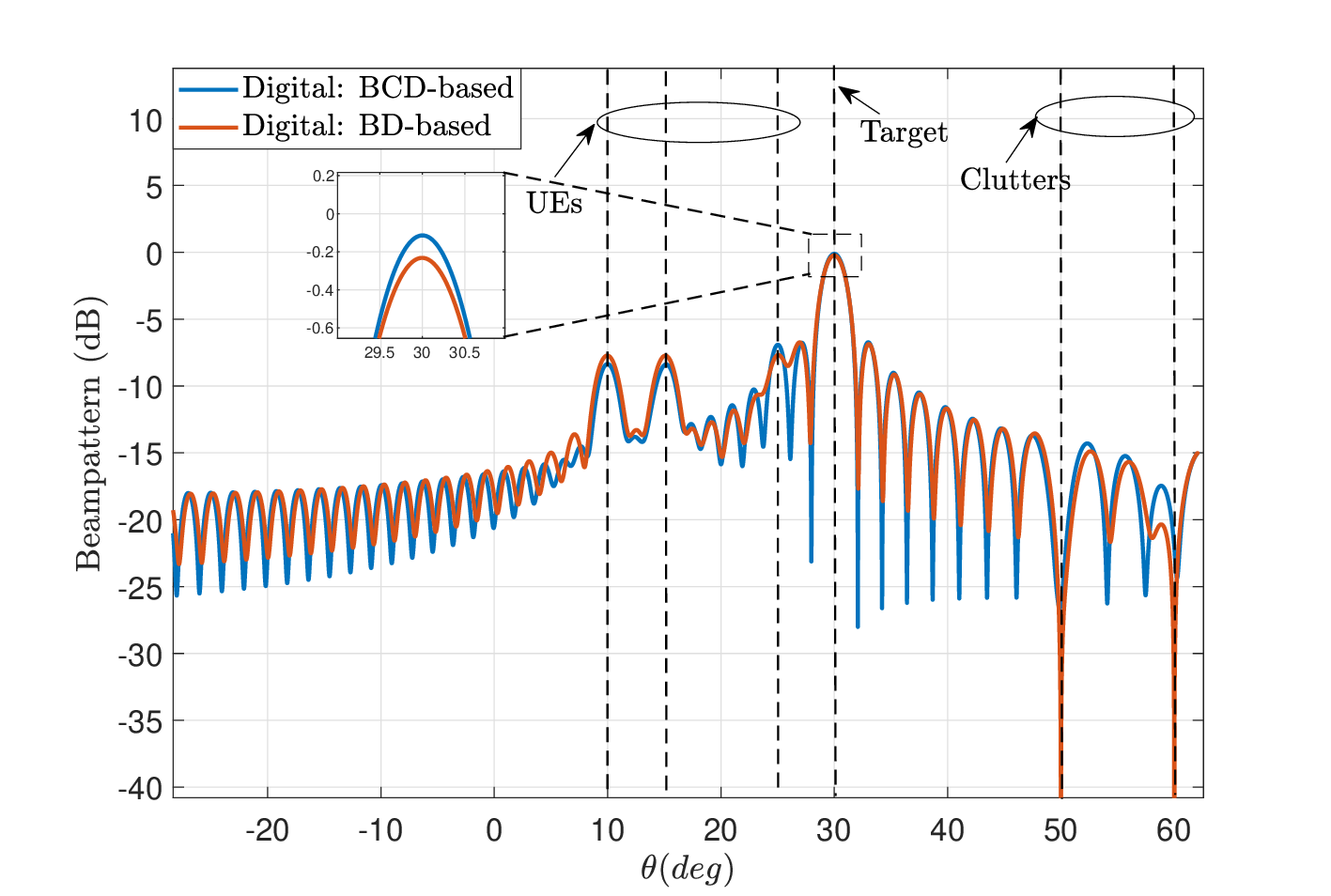}
  \end{minipage}
  \label{FIG7_4}
 }
  \caption{Transmit beam pattern for different UEs.}
 \label{fig_pattern}
\end{figure*}




In Fig. \ref{fig_Pt}(a), we plot the overall communication-sensing
performance as a function of the transmit power $P_t$, where the
weighting coefficient $\eta$ is set to $0.5$. Also, we depict the
WSR and SCNR in Fig. \ref{fig_Pt}(b) and Fig. \ref{fig_Pt}(c),
respectively. We see that, as expected, the performance increases
as the transmit power $P_t$ grows. Also, it is observed that the
performance of using the hybrid precoder/combiner is close to the
performance attained by fully digital precoder/combiner.
Furthermore, the proposed BCD-based method can deliver a better
communication/sensing performance than the BD-based solution, and
the performance improvement is more pronounced as the transmit
power increases.

In Fig. \ref{fig_NtRF}, we depict the overall
communication-sensing performance of the proposed BCD method
versus the number of RF chains $N_t^{\rm RF}$, where the transmit
power is set to $P_t = 30$ dBm. To better illustrate the
performance, the weighting coefficient is set to $\eta = 1$ and
$\eta = 0.55$, respectively. The total number of resolvable paths
at the BS is set to $r=18$, in which the number of clutter patches
is set to $I=2$. To illustrate how many RF chains are required,
the performance attained by the BCD method with a fully digital
precoder is also included for a comparison. We see that when $\eta
= 1$, which implies only communication performance is concerned,
the BCD method with a hybrid precoder achieves the same
performance as that of a fully digital precoder using only
$N_t^{\rm RF}=r-I-1=15$ RF chains. When $\eta=0.55$, the BCD
method with a hybrid precoder requires $\tilde{r}=r-I=16$ RF
chains to achieve a performance similar to that of a fully digital
precoder. This result is attributed to the fact that the radar
receive beamforming vector $\boldsymbol{w}$ has effectively
eliminate the interference caused by the clutter patches. These
results corroborate our theoretical result in Proposition
\ref{proposition-hybrid-suff} which states that $\tilde{r}$ RF
chains are sufficient to achieve the same performance as that of a
fully digital precoder.

Fig. \ref{fig_region} illustrates the tradeoff between the
communication and sensing performance by varying the weighting
coefficient $\eta$ in $[0,1]$, where the transmit power is set to
$P_t = 20$ and $P_t = 30$ dBm, respectively. From Fig.
\ref{fig_region}, we see that by choosing a proper weighting
coefficient $\eta^{\ast}$, a good balance between the
communication and sensing performance can be achieved.
Specifically, for such a value of $\eta^{\ast}$, both
communication and sensing achieve a decent performance that incurs
only a mild performance loss (less than $20\%$) as compared with
the performance attained by optimizing exclusively for a single
task. Also, a further increase (resp. decrease) of $\eta^{\ast}$
only leads to a small improvement in communication (resp. sensing)
performance, but results in a substantial sensing (resp.
communication) performance degradation.



\subsection{Beam Pattern Analysis}
To gain insight into the optimized precoder, we examine their
transmit beam patterns. In our simulations, we set $\eta = 0.5$,
$K=3$ and $N_s=1$. For simplicity of illustration, the directions
of UEs, target and the clutter patches are fixed as $\{
10^{\circ}, 15^{\circ}, 25^{\circ}\}$, $30^{\circ}$, and $\{
50^{\circ}, 60^{\circ} \}$, respectively.
Fig.\ref{fig_pattern}(a)-Fig.\ref{fig_pattern}(c) depict the beam
pattern for each UE, in which the beam pattern for the $k$th UE is
defined as $P_k(\theta) = \| \boldsymbol{a}_t^H(\phi_0)
\boldsymbol{F}_k \|^2$. The transmit power is set to $P_t = 30$
dBm.

It can be observed in Fig. \ref{fig_pattern} that for both
methods, the optimized radiation pattern forms directional beams
towards the target and the served user, and creates nulls in
directions pointing to other users as well as the clutter patches.
Such a beam pattern can effectively increase the desired
communication/radar signal while suppressing interference caused
by other users and clutter patches, which is beneficial for
improving both sensing and communication performance. We also
observe that, despite similar beam patterns, these two methods are
slightly different in their power allocation strategies. For the
BCD-based method, it tends to allocate more power to the user that
is near to the target. Apparently, assigning more power to this
user not only helps increase the communication performance, but
also improve the sensing accuracy. In contrast, the BD-based
analytical solution tends to allocate nearly the same power to
each UE. This phenomenon is consistent with the results reported
in Proposition \ref{proposition-power}.

\section{Conclusion}\label{sec-conclu}
In this paper, we studied the problem of joint transceiver design
for mmWave/THz ISAC systems. Such a problem was formulated into a
non-convex optimization problem whose objective is to maximize the
sum of all communication users' rates and the radar's received
SCNR. By exploring the low-dimensional subspace property of the
optimal precoder, we developed a computationally efficient
BCD-based algorithm for joint transceiver design. In addition, by
generalizing the BD idea to the ISAC system, we proposed an
analytical solution to the joint transceiver design problem. Simulation
results were provided to illustrate the effectiveness of the
proposed methods. Specifically, we showed that by choosing a
proper weighting coefficient, the communication and sensing
performance can be well balanced, with the performance of each
task incurring only a mild performance degradation as compared
with the performance attained by  optimizing exclusively for a
single task.


\appendices
\section{Proof of Proposition \ref{proposition1}}
\label{appendixA}

\begin{figure*}[!t]
\normalsize
\setcounter{MYtempeqncnt}{\value{equation}}
\setcounter{equation}{76}
\begin{align}
\label{eqn-big}
& \sum_{k=1}^K \sum_{i=1}^K \rho_c\alpha_{k}(\boldsymbol{F}_i^{\star})^H
\boldsymbol{H}_k^H \boldsymbol{W}_k \bigg( \sum_{j=1}^K\boldsymbol{W}_k^H
\boldsymbol{H}_k \boldsymbol{F}_j^{\star}(\boldsymbol{F}_j^{\star})^H \boldsymbol{H}_k^H
\boldsymbol{W}_k + \sigma_k^2 \boldsymbol{W}_k^H\boldsymbol{W}_k  \bigg)^{-1}
\boldsymbol{W}_k^H  \boldsymbol{H}_k\boldsymbol{F}_i^{\star} + \frac{\rho_s }{f_2^2}
\sum_{i=1}^K f_2 {\rm tr} \left((\boldsymbol{F}_i^{\star})^H
\boldsymbol{A}^H\boldsymbol{w}\boldsymbol{w}^H \boldsymbol{A} \boldsymbol{F}_i^{\star}  \right) \nonumber \\
=& \sum_{k=1}^K  \sum_{i\neq k}^K \rho_c \alpha_{k}(\boldsymbol{F}_i^{\star})^H
\boldsymbol{H}_k^H \boldsymbol{W}_k \bigg( \sum_{j\neq k}^K\boldsymbol{W}_k^H
\boldsymbol{H}_k \boldsymbol{F}_j^{\star}(\boldsymbol{F}_j^{\star})^H \boldsymbol{H}_k^H \boldsymbol{W}_k
+ \sigma_k^2 \boldsymbol{W}_k^H\boldsymbol{W}_k  \bigg)^{-1}  \boldsymbol{W}_k^H
\boldsymbol{H}_k\boldsymbol{F}_i^{\star} \nonumber \\
&+   \frac{\rho_s }{f_2^2} \sum_{i=1}^K f_{1,i} {\rm tr} \left((\boldsymbol{F}_i^{\star})^H
\sum_{m=1}^I\boldsymbol{B}_m^H\boldsymbol{w}\boldsymbol{w}^H \boldsymbol{B} _m\boldsymbol{F}_i^{\star} \right)
\end{align}
\begin{align}
\label{eq-tr-fi}
&\sum_{k=1}^K \rho_c w_k \sigma_k^2 {\rm tr}\bigg(  (\boldsymbol{W}_k^H\boldsymbol{W}_k )^{-1}
\big(\sum_{j=1}^K\boldsymbol{W}_k^H \boldsymbol{H}_k \boldsymbol{F}_j^{\star}(\boldsymbol{F}_j^{\star})^H
\boldsymbol{H}_k^H \boldsymbol{W}_k + \sigma_k^2 \boldsymbol{W}_k^H\boldsymbol{W}_k  \big)^{-1} \bigg)
+\frac{\rho_s }{f_2^2} \sum_{i=1}^K f_{1,i} \sum_{m=1}^M{\rm tr} \left((\boldsymbol{F}_i^{\star})^H
\boldsymbol{B}_m^H\boldsymbol{w}\boldsymbol{w}^H \boldsymbol{B}_m \boldsymbol{F}_i^{\star} \right)  \nonumber \\
=& \sum_{k=1}^K \rho_c w_k \sigma_k^2 {\rm tr}\bigg(
(\boldsymbol{W}_k^H\boldsymbol{W}_k )^{-1} \big(\sum_{j\neq
k}^K\boldsymbol{W}_k^H \boldsymbol{H}_k
\boldsymbol{F}_j^{\star}(\boldsymbol{F}_j^{\star})^H
\boldsymbol{H}_k^H \boldsymbol{W}_k + \sigma_k^2
\boldsymbol{W}_k^H\boldsymbol{W}_k  \big)^{-1} \bigg)
  +   \frac{\rho_s }{f_2^2} \sum_{i=1}^K f_2 {\rm tr} \left((\boldsymbol{F}_i^{\star})^H
\boldsymbol{A}^H\boldsymbol{w}\boldsymbol{w}^H \boldsymbol{A} \boldsymbol{F}_i^{\star}  \right)
\end{align}
\setcounter{equation}{\value{MYtempeqncnt}}
\hrulefill
\vspace*{4pt}
\end{figure*}

To demonstrate the low-dimensional subspace property, we first
introduce the definition of trivial points and then analyze the
KKT conditions of problem \eqref{eq-opt-simp}. Specifically, if a
point $\{ \boldsymbol{F}_k \}_{k=1}^K$ satisfying
$\boldsymbol{H}_k \boldsymbol{F}_k = \boldsymbol{0}$ and
$\boldsymbol{A} \boldsymbol{F}_k = \boldsymbol{0}$, which results
in a zero WSR and a zero SCNR, we say it is a trivial point of
problem \eqref{eq-opt-simp}. Intuitively, any effective optimized
solution to problem \eqref{eq-opt-simp} should be a non-trivial
KKT point. Then, we have the following proposition regarding the
dual variable $\lambda^{\star}$.
\begin{Proposition}
For any non-trivial KKT point of problem \eqref{eq-opt-simp}, the
dual variable $\lambda^{\star}$ associated with the transmit power
constraint must be positive, i.e., $\lambda^{\star} > 0$.
\label{propo-lamb}
\end{Proposition}
\begin{proof}
Denote $\{ \boldsymbol{F}_k ^{\star}\}_{k=1}^K$ as the KKT point
of the problem \eqref{eq-opt-simp}, which satisfies
\begin{align}
&\rho_c w_k \nabla_{\boldsymbol{F}_k^{\star}} R_k(\boldsymbol{F}^{\star})+
\rho_c\sum_{i \neq k}^K  w_i \nabla_{\boldsymbol{F}_k^{\star}} R_i(\boldsymbol{F}^{\star})  \nonumber \\
&+\rho_s\nabla_{\boldsymbol{F}_k^{\star}} \text{SCNR}(\boldsymbol{F}^{\star})
-\lambda^{\star} \boldsymbol{F}_k^{\star}=\boldsymbol{0}, \forall k, \label{condi-first} \\
& \left(\sum_{k=1}^K \operatorname{Tr}\left(\boldsymbol{F}_k^{\star}
\left(\boldsymbol{F}_k^{\star}\right)^H\right)-P_{t}\right) \cdot \lambda^{\star}=0, \label{condi-2nd}\\
& \sum_{k=1}^K \operatorname{Tr}\left(\boldsymbol{F}_k^{\star}
\left(\boldsymbol{F}_k^{\star}\right)^H\right) \leq P_{\max }, \\
& \lambda^{\star} \geq 0,
\end{align}
where $R_k(\boldsymbol{F}^{\star})$ and
$\text{SCNR}(\boldsymbol{F}^{\star})$ are defined in \eqref{eq-Rk}
and \eqref{eq-SCNR} in terms of $\boldsymbol{F}^{\star} =
[\boldsymbol{F}_1^{\star} \phantom{0} \ldots \phantom{0}
\boldsymbol{F}_K^{\star}] \in \mathbb C^{N_t \times KN_s}$,
respectively.

We prove $\lambda^{\star} > 0$ by contradiction. We first assume
$\lambda^{\star} = 0$. Note that the gradient of
$\nabla_{\boldsymbol{F}_k} R_k(\boldsymbol{F}^{\star})$,
$\nabla_{\boldsymbol{F}_k} R_i(\boldsymbol{F}^{\star})$, and
$\nabla_{\boldsymbol{F}_k^{\star}} \text{SCNR}
(\boldsymbol{F}^{\star})$ can be respectively calculated as
\begin{align}
\nabla_{\boldsymbol{F}_k^{\star}} R_k(\boldsymbol{F}^{\star}) =&
\boldsymbol{H}_k^H \boldsymbol{W}_k \boldsymbol{Z}_{k}^{-1}  \boldsymbol{W}_k^H
\boldsymbol{H}_k\boldsymbol{F}_k^{\star}, \label{grad-Rk} \\
 \nabla_{\boldsymbol{F}_k^{\star}} R_i(\boldsymbol{F}^{\star}) =& \boldsymbol{H}_i^H
 \boldsymbol{W}_i (\boldsymbol{Z}_i^{-1}  -  \tilde{\boldsymbol{Z}}_i^{-1} )
 \boldsymbol{W}_i^H  \boldsymbol{H}_i\boldsymbol{F}_k^{\star} ,\\
   \nabla_{\boldsymbol{F}_k^{\star}} \text{SCNR}(\boldsymbol{F}^{\star}) =&
   \frac{1 }{f_2^2} ( f_2 \boldsymbol{A}^H\boldsymbol{w}\boldsymbol{w}^H
   \boldsymbol{A} \boldsymbol{F}_k^{\star} \nonumber \\
    & - f_{1,k} \sum_{i=1}^I\boldsymbol{B}_i^H\boldsymbol{w}\boldsymbol{w}^H
    \boldsymbol{B}_i \boldsymbol{F}_k^{\star}), \label{grad-Fk}
\end{align}
where
\begin{align}
\boldsymbol{Z}_k =&  \sum_{j=1}^K\boldsymbol{W}_k^H
 \boldsymbol{H}_k \boldsymbol{F}_j^{\star}(\boldsymbol{F}_j^{\star})^H \boldsymbol{H}_k^H \boldsymbol{W}_k +
 \sigma_k^2 \boldsymbol{W}_k^H\boldsymbol{W}_k, \forall k,  \nonumber \\
 \tilde{\boldsymbol{Z}}_i =&
 \sum_{j\neq i}^K\boldsymbol{W}_i^H \boldsymbol{H}_i
\boldsymbol{F}_j^{\star}(\boldsymbol{F}_j^{\star})^H \boldsymbol{H}_i^H
\boldsymbol{W}_i + \sigma_i^2 \boldsymbol{W}_i^H\boldsymbol{W}_i, \forall i, \nonumber
\end{align}

$f_{1,k} = \boldsymbol{w}^H \boldsymbol{A} \boldsymbol{F}_k^{\star}
(\boldsymbol{F}_k^{\star})^H \boldsymbol{A}^H \boldsymbol{w}$ and $f_2 =
\boldsymbol{w}^H
(\sum_{i=1}^I\boldsymbol{B}_i\boldsymbol{F}^{\star}(\boldsymbol{F}^{\star})^H
\boldsymbol{B}_i^H +\sigma^2 \boldsymbol{I})\boldsymbol{w}$.
Left-multiplying the first-order optimality condition in
\eqref{condi-first} by $(\boldsymbol{F}_k^{\star})^H$, we arrive
at
\begin{align}
    &\rho_c w_k (\boldsymbol{F}_k^{\star})^H\nabla_{\boldsymbol{F}_k^{\star}}
    R_k(\boldsymbol{F}^{\star})+\rho_c \sum_{i \neq k}^K w_i (\boldsymbol{F}_k^{\star})^H
    \nabla_{\boldsymbol{F}_k^{\star}} R_i(\boldsymbol{F}^{\star}) \nonumber \\
 &+\rho_s(\boldsymbol{F}_k^{\star})^H \nabla_{\boldsymbol{F}_k^{\star}}
 \text{SCNR}(\boldsymbol{F}^{\star}) = \boldsymbol{0}. \label{eq-grad-sing}
\end{align}
Taking summation over all $k=1,2,\ldots, K$ gradients in
\eqref{eq-grad-sing} and re-arranging the terms, we could obtain
\eqref{eqn-big} at the top of this page. Using the identities
${\rm tr}(\boldsymbol{A}\boldsymbol{B}) = {\rm tr}(\boldsymbol{B}
\boldsymbol{A})$ and ${\rm
tr}(\boldsymbol{A}(\boldsymbol{A}+\boldsymbol{I})^{-1}) = {\rm
tr}( \boldsymbol{I}) - {\rm
tr}(\boldsymbol{A}+\boldsymbol{I})^{-1} $, we can further obtain
\eqref{eq-tr-fi}.

Comparing each term in the left-hand-side and right-hand-side of
\eqref{eq-tr-fi}, it can be seen that the equality holds iff
$\boldsymbol{H}_k\boldsymbol{F}_k^{\star}=\boldsymbol{0},\forall k$ and
$\boldsymbol{A}\boldsymbol{F}_k^{\star} = \boldsymbol{0},\forall k$, which
contradicts the fact that $\boldsymbol{F}_k^{\star}$ is a
non-trivial point. Therefore, we conclude that for any given
non-trivial point, the corresponding dual variable
$\lambda^{\star} >0$. Also, the full power property can thus be
deduced by checking the slack-complementary condition in
\eqref{condi-2nd}.

\setcounter{equation}{78} Given $\lambda^{\star}
>0$, we then prove that the corresponding
$\boldsymbol{F}_k^{\star}$ must be a non-trivial point. Similarly,
we prove it by contradiction. If
$\boldsymbol{H}_k\boldsymbol{F}_k^{\star}=\boldsymbol{0},\forall
k$ and $\boldsymbol{A}\boldsymbol{F}_k^{\star} =
\boldsymbol{0},\forall k$, the gradients in \eqref{grad-Rk} to
\eqref{grad-Fk} are equal to zero. In this case, the
first-order optimality condition in \eqref{condi-first} reduces to
$\lambda^{\star} \boldsymbol{F}_k^{\star} = \boldsymbol{0}$, which
implies $\boldsymbol{F}_k^{\star} = \boldsymbol{0}$. This
contradicts \eqref{condi-2nd}. Therefore, we can conclude
that for a non-trivial point
$\{\boldsymbol{F}_k^{\star}\}_{k=1}^K$, the corresponding dual
variable $\lambda^{\star} > 0$.
\end{proof}

Looking back on
\eqref{condi-first} and noting $\lambda^{\star} >0$, the precoder $\boldsymbol{F}_k^{\star}$ can
be represented as
\begin{align}
\boldsymbol{F}_k^{\star} =& \frac{1}{\lambda} \bigg(\rho_c w_k
\nabla_{\boldsymbol{F}_k^{\star}} R_k(\boldsymbol{F}^{\star}) +
\sum_{i \neq k}^K \rho_c w_i \nabla_{\boldsymbol{F}_k^{\star}} R_i(\boldsymbol{F}^{\star}) \nonumber \\
 & +\rho_s\nabla_{\boldsymbol{F}_k^{\star}} {\text{SCNR}}(\boldsymbol{F}^{\star})  \bigg)
 =\boldsymbol{H}^H \boldsymbol{G}_k
\end{align}
with
\begin{align}
   & \boldsymbol{H}\triangleq [\boldsymbol{H}_1^T\phantom{0}
\ldots\phantom{0} \boldsymbol{H}_K^T\phantom{0}
\boldsymbol{A}^T\phantom{0} \boldsymbol{B}_1^T\phantom{0}
\ldots\phantom{0}\boldsymbol{B}_{I}^T]^T \in \mathbb C^{M \times
N_t}, \nonumber\\
 &\boldsymbol{G}_k  \in \mathbb C^{M \times N_{s}} ,    \quad M = KN_r + (I+1)N_{\rm sen}, \nonumber
\end{align}
It is seen that $\boldsymbol{F}_k^{\star}$ must lie in the range
space of $\boldsymbol{H}^H$. Let $\boldsymbol{V} \in \mathbb
C^{N_t \times r} $ denote the matrix comprising the steering
vectors from the BS to all $K$ UEs, the target and all clutters
with $r= \sum_{k=1}^K L_k + I +1 $, it is easy to verify that
$\boldsymbol{V}$ and $\boldsymbol{H}^H$ have the same range
space\footnote{If either $\rho_c$ or $\rho_s$ is zero, the corresponding paths are irrelevant for either communication or
sensing. Therefore, we can simply remove them from the matrix
$\boldsymbol{V}$. This reduces the dimension of $\boldsymbol{V}$
to either $I+1$ (if only communication paths are irrelevant) or
$r-I-1$ (if only sensing paths are irrelevant).}. In other words,
the digital precoder $\boldsymbol{F}_k^{\star} , \forall k$ can be
represented by
       \begin{align}
        \boldsymbol{F}_k^{\star} = \boldsymbol{V} \boldsymbol{X}_k, \label{Fk-Xk}
    \end{align}
This completes the proof.

\section{Proof of Proposition \ref{proposition-scaling}}
\label{appendix-scaling}
Firstly, it is easy to see that $(\{\boldsymbol{W}_k^{\rm
opt_2},\alpha\boldsymbol{X}_k^{\rm opt_2} \}_{k=1}^K,
\boldsymbol{w}^{\rm opt_2})$ satisfies the transmit power
constraint in \eqref{eq-opt-X}, which shows it is a feasible
solution of \eqref{eq-opt-X}. Next, we show that
$(\{\boldsymbol{W}_k^{\rm opt_2},\alpha^{}
\boldsymbol{X}_k^{\rm opt_2} \}_{k=1}^K, \boldsymbol{w}^{\rm
opt_2})$ attains the maximum objective function value in
\eqref{eq-opt-X}. Substituting $(\{\boldsymbol{W}_k^{\rm
opt_2},\alpha^{} \boldsymbol{X}_k^{\rm opt_2} \}_{k=1}^K,
\boldsymbol{w}^{\rm opt_2})$ into the objective function of
\eqref{eq-opt-X}, we arrive at
\begin{align}
& \rho_c R( \alpha^{} \boldsymbol{X}_k^{\rm opt_2} ,\boldsymbol{W}_k^{\rm opt_2})
+ \rho_s \text{SCNR}  (\boldsymbol{w}^{\rm opt_2}, \alpha \boldsymbol{X}_k^{\rm opt_2}) ,\nonumber \\
=&  \rho_c \bar{R}(  \boldsymbol{X}_k^{\rm opt_2} ,\boldsymbol{W}_k^{\rm opt_2})
+ \rho_s \overline{\text{SCNR}}  (\boldsymbol{w}^{\rm opt_2}, \boldsymbol{X}_k^{\rm opt_2}), \nonumber \\
\geq & \rho_c\bar{R}(  \boldsymbol{X}_k ,\boldsymbol{W}_k) +
\rho_s \overline{\text{SCNR}}  (\boldsymbol{w}^{}, \boldsymbol{X}_k^{}),  \nonumber \\
\stackrel{(a)}=& \rho_c {R}( \tilde{\alpha} \boldsymbol{X}_k
,\boldsymbol{W}_k) + \rho_s {\text{SCNR}}  (\boldsymbol{w}^{},
\tilde{\alpha} \boldsymbol{X}_k^{}) ,
\end{align}
where $\tilde{\alpha}$ is a scaling factor to satisfy the transmit
power constraint and is defined as
\begin{align}
\tilde{\alpha} = \sqrt{ \frac{P_t}{\sum_{k=1}^K {\rm
tr}(\tilde{\boldsymbol{H}}\boldsymbol{X}_k (\boldsymbol{X}_k^H)
}   }.
\end{align}
Hence, the point $(\{\boldsymbol{W}_k^{\rm opt_2}, \alpha^{}
\boldsymbol{X}_k^{\rm opt_2} \}_{k=1}^K, \boldsymbol{w}^{\rm
opt_2})$ is the optimal solution to \eqref{eq-opt-X}. This
completes the proof.

\bibliography{newbib}
\bibliographystyle{IEEEtran}

\end{document}